\documentclass[reprint,superscriptaddress,
 amsmath,amssymb,
 aps,
pra
]{revtex4-2}

\usepackage{graphicx}
\usepackage{dcolumn}
\usepackage{array}
\usepackage{physics}
\usepackage{amsfonts, amsmath}
\usepackage{dsfont}
\usepackage{amsthm}
\usepackage{bm}
\usepackage{hyperref}
\usepackage{bbold}
\usepackage{color}
\usepackage[english]{babel}
\usepackage[normalem]{ulem}
\usepackage{tabularx}

\newcommand{\T}[0]{{\mathrm{T}}}

\newcommand{\poly}{\mathrm{poly}}
\newtheorem{theorem}{Theorem}
\newtheorem{corollary}{Corollary}
\newtheorem{proposition}{Proposition}

\newtheorem{lemma}{Lemma}

\usepackage{mathtools}
\newcommand{\Per}{\operatorname{Per}}
\newcommand{\TV}{\operatorname{TV}}
\newcommand{\im}{\operatorname{im}}

\newcommand{\Nop}{\hat N}
\newcommand{\fall}[2]{(#1)_{\underline{#2}}}
\renewcommand{\ket}[1]{\lvert #1\rangle}
\renewcommand{\bra}[1]{\langle #1\rvert}

\renewcommand{\norm}[1]{\left\lVert #1\right\rVert}
\renewcommand{\abs}[1]{\left\lvert #1\right\rvert}
\providecommand{\diag}{\operatorname{diag}}

\allowdisplaybreaks
\begin{document}
\definecolor{red}{RGB}{255,0,0}

\title{Logarithmic-depth quantum simulation of boson sampling}
\author{Changhun Oh}
\email{changhun0218@gmail.com}
\affiliation{Department of Physics, Korea Advanced Institute of Science and Technology, Daejeon 34141, Korea}
\begin{abstract}
We show that boson sampling with an arbitrary $m$-mode interferometer and $n\le m$ single-photon inputs can be simulated to inverse-polynomial total-variation error by a logarithmic-depth qubit circuit with polynomially many qubits. The circuit uses Clifford+$T$ gates, arbitrary qubit connectivity, and a single final measurement, and its family is logspace uniform. The key idea is to enlarge the optical system, decompose the resulting transformation into six quadratic shears, and distribute each mode over many submodes. This redistribution permits a fixed local occupation cutoff, after which local basis changes and parallel phase gates give the qubit circuit. Consequently, our result places boson sampling within shallow quantum computation.
\end{abstract}
\maketitle

\section{Introduction}
Boson sampling is a restricted model of quantum computation whose output distribution is believed to be hard to sample classically~\cite{AA,Hamilton,GrierGBS}. Its implementation with noninteracting photons and passive linear optics has made it a leading proposal for demonstrating quantum advantage. Experimental realizations have advanced substantially, including the Jiuzhang series~\cite{Zhong,Zhong2021,Deng2023,Liu2026}, Borealis~\cite{Madsen2022}, and an atomic boson sampler~\cite{Young2024}. Alongside these experiments, extensive theoretical work has developed classical simulation algorithms and clarified how loss, partial distinguishability, noise, and circuit structure govern classical complexity~\cite{CC,CCfast,Neville,Quesada2020,Quesada2022,Bulmer2022,OhMPO2021,LiuMPO2023,OhNoisy,OhGBS,OhGraph,OhConstant,GoPartial2025}. These developments are surveyed in Refs.~\cite{HangleiterEisert,OhReview}.

Despite this progress, the quantum complexity of boson sampling is not fully understood. Standard boson sampling is widely believed to be nonuniversal, but a strict separation from universal quantum computation has not been established~\cite{AA,HangleiterEisert}. More broadly, boson sampling is built from passive linear optics, which can support universal quantum computation when supplemented with ancillary photons, adaptive measurements, and feedforward~\cite{KLM}. Boson sampling omits these intermediate controls, making it natural to ask which quantum resources suffice to simulate its output distributions and how those requirements compare with general quantum computation.

A natural way to address this question is to examine the depth of qubit circuits needed to simulate boson sampling. Depth counts sequential layers of one- and two-qubit gates, measuring how much quantum processing must remain sequential even when independent operations are performed in parallel. With comparable gate times, it also sets the duration over which quantum coherence must be maintained. Qubit encodings already provide a route to simulating bosonic systems on quantum computers~\cite{Sawaya,Tong}. Prior qubit-circuit constructions establish polynomial-size simulation of boson sampling~\cite{Moylett} and implement arbitrary linear interferometers through decompositions into two-mode elements, with depth polynomial in the numbers of modes and photons~\cite{Leone}. They therefore leave open how far the sequential depth can be reduced. Since additional qubits can enable parallel operations~\cite{MN,MNcodes,HS}, the key question is whether polynomially many qubits suffice to simulate arbitrary boson sampling in logarithmic depth. Here, the depth refers to the qubit simulator, rather than the optical interferometer.

We show that logarithmic depth suffices with a polynomial number of qubits. For an arbitrary $m$-mode interferometer, one photon in each of the first $n\le m$ modes and vacuum elsewhere, and target total-variation error $\epsilon$, our sampler has depth $O(\log(m/\epsilon))$ and width polynomial in $m$ and $1/\epsilon$. It uses Clifford+$T$ gates on arbitrarily connected qubits and a single final measurement. For every fixed inverse-polynomial accuracy, the circuit family is logspace uniform; moreover, its width can be made $O(m^\alpha)$ for any fixed $\alpha>2$. The guarantee applies to the boson-sampling distribution, including all photon collisions. It therefore gives an upper bound on the quantum depth of standard boson sampling independently of the conjectures supporting its classical sampling hardness.

The key is to reorganize the optical transformation before encoding it in qubits. We first add vacuum modes to obtain an equivalent sampling circuit on $4m$ modes that factors exactly into six quadratic shears. Each shear involves mutually commuting position or momentum quadratures. We then distribute each mode uniformly over many submodes, reducing the occupation per submode so that a fixed local cutoff permits a representation using a constant number of qubits~\cite{Sawaya}. Summing the output counts recovers the original sampling distribution. Local basis changes, coherent counting, and parallel phase gates implement the truncated shears in logarithmic depth~\cite{Gossett,MNcodes}, with this scaling preserved over a finite gate set~\cite{KMM,Ross}. We control the truncation error throughout the evolution, including the photons temporarily created by the individual shears.

With the interferometer entries supplied as binary input bits, the sampler also places the associated promise decision problems for logspace-uniform $\mathrm{NC}^1$ output predicates with a constant probability gap in $\mathrm{BQNC}^1$, as formulated in Section~\ref{sec:decision}. The classical sampling hardness that motivates boson-sampling experiments is therefore compatible with a shallow qubit implementation. This upper bound identifies a restricted quantum circuit model sufficient for the simulation, while leaving open the broader question of universality and whether still less depth can suffice.

\section{Boson sampling and main result}\label{sec:model}
In this section, we fix the input model, circuit conventions, and target distribution, and then state the depth and width bounds proved below.

Let $U\in U(m)$ describe a passive interferometer, with the mode transformation $\hat a_j^\dagger\mapsto\sum_{i=1}^m U_{ij}\hat a_i^\dagger$. The mode operators satisfy $[\hat a_i,\hat a_j^\dagger]=\delta_{ij}$. We use quadratures $\hat q_j\coloneqq(\hat a_j+\hat a_j^\dagger)/\sqrt2$ and $\hat p_j\coloneqq(\hat a_j-\hat a_j^\dagger)/(i\sqrt2)$, so $[\hat q_i,\hat p_j]=i\delta_{ij}$ and $[\hat q_i,\hat q_j]=[\hat p_i,\hat p_j]=0$. These conventions also apply to the auxiliary bosonic modes introduced below. Throughout, hats denote operators on Fock space, whereas matrices such as $U$ act on mode labels.

The input consists of one photon in each of the first $n\le m$ modes and vacuum in the remaining modes, with occupation vector $\bm t\coloneqq(1,\ldots,1,0,\ldots,0)$. For an output occupation pattern $\bm n\coloneqq(n_1,\ldots,n_m)$, with total photon number $n=\sum_{i=1}^m n_i$, the target distribution is
\begin{align}
P_U(\bm n)\coloneqq\frac{|\Per(U_{\bm n,\bm t})|^2}{\prod_{i=1}^m n_i!},\qquad \sum_{i=1}^m n_i=n.\label{eq:bsprob}
\end{align}
Here, $U_{\bm n,\bm t}$ repeats output row $i$ according to $n_i$ and selects the occupied input columns. We allow a failure symbol $\perp$ in addition to valid $n$-photon occupation patterns and extend the ideal distribution by $P_U(\perp)=0$. We use $\|P-Q\|_{\TV}=\tfrac12\sum_x|P(x)-Q(x)|$, where the sum includes $\perp$. All sampling guarantees are unconditional. We write $\|A\|$ for the operator norm.

To formulate the simulation task in the qubit circuit model, we specify how the interferometer is presented to the circuit. An input instance contains the photon number $n\le m$ in binary and two's-complement fixed-point descriptions of every real and imaginary matrix entry, with a fixed number of integer bits sufficient to represent $[-1,1]$. If each real component has $b$ bits, the decoded matrix $\widetilde U$ has an explicit description of length $N=\Theta(m^2b)$. Every input bit occupies its own computational basis wire, giving direct access to the matrix-entry bits. The matrix entries are promised to approximate those of a unitary $U$ to the precision specified in Section~\ref{sec:resources}.

Circuit depth counts layers of disjoint one- and two-qubit gates, with arbitrary connectivity. Width counts all qubits, including the matrix-entry input registers and all initialized ancillas. We take logarithms in resource bounds to be at least one. We use the fixed finite gate set Clifford+$T$ so that the depth bound includes the cost of approximating the required rotations, rather than treating arbitrary continuous gates as exact operations of unit cost~\cite{KMM,Ross}. This places the result directly in a finite gate model of quantum circuit complexity~\cite{NishimuraOzawa}. The circuit ends with a computational basis measurement; its output occupations can be computed coherently beforehand within the same depth order. No geometric locality is assumed.

With the input and circuit model fixed, we can now state the main result.

\begin{theorem}[Boson sampling in logarithmic quantum depth]\label{thm:approx}
For every $m$-mode unitary $U$ supplied through the finite-precision input model above, every photon number $n\le m$ with one photon in each of the first $n$ modes and vacuum elsewhere, and every $0<\epsilon<1$, there is a Clifford+$T$ circuit of depth $O(\log(m/\epsilon))$ and width polynomial in $m$ and $1/\epsilon$ whose output distribution $\widetilde P_U$ satisfies
\begin{align}
\|\widetilde P_U-P_U\|_{\TV}\le\epsilon.
\end{align}
For every fixed inverse-polynomial accuracy, these circuits form a logspace-uniform family of depth $O(\log m)$.
\end{theorem}

The width can moreover be made arbitrarily close to quadratic in the number of modes.

\begin{corollary}[Width arbitrarily close to quadratic]\label{cor:width}
For any fixed inverse-polynomial accuracy and any fixed exponent $\alpha>2$, the sampler in Theorem~\ref{thm:approx} can be implemented by a logspace-uniform qubit circuit of depth $O(\log m)$ and width $O(m^\alpha)$.
\end{corollary}
The construction parameters and the width--accuracy tradeoff are derived in Section~\ref{sec:error}.

To place the construction within the standard uniform-circuit framework, we treat the entries of $U$ as input bits rather than compiling a separate circuit for each interferometer. For fixed size and accuracy parameters, a single circuit layout works for every allowed setting of these bits. The six shear coefficients are fixed constants or signed copies of the real and imaginary parts of $U$, and CNOT trees distribute coefficient bits wherever they are needed in logarithmic depth. Thus, the circuit performs no instance-dependent diagonalization or numerical decomposition of $U$: the local basis change in Step 3 depends only on the fixed cutoff $s$, while $U$ enters through input-controlled rotations whose angles are determined by bit positions and circuit indices.

Beyond being independent of the interferometer entries, this circuit layout can be generated in logarithmic space. Given the size parameters and the prescribed inverse-polynomial accuracy, a classical logspace machine generates the wiring and the Clifford+$T$ descriptions of the fixed local rotations, while the numerical entries of $U$ are supplied through the input wires. Appendix~\ref{sec:uniformity} gives the generator explicitly. For fixed inverse-polynomial accuracy, the stated precision gives $b=O(\log m)$. Consequently, the input length $N=\Theta(m^2b)$ satisfies $\log N=\Theta(\log m)$.

\subsection{Mapping from optical modes to qubits and proof structure}
We now outline the construction underlying the theorem and explain the role of each step. Figure~\ref{fig:construction} illustrates the mapping, including the input and readout.

\begin{enumerate}
\item \emph{Enlargement and shear decomposition.} Add $3m$ vacuum modes and construct an equivalent passive interferometer on $4m$ modes whose photon counts, summed over each output pair, reproduce the target distribution. It consists of a fixed passive input encoder followed by a passive transformation with mode matrix $R_U$, which is real and symmetric and satisfies $R_U^2=I$. Both transformations admit exact decompositions into three shears, up to phases that do not affect photon counting, giving six shears in total.
\item \emph{Redistribution over submodes.} Append $K-1$ vacuum submodes per mode and apply uniform splitters before the six shears. Transform each shear with the same splitters. The resulting $4mK$-mode circuit has the original photon-counting distribution after summing the submode counts within each output pair.
\item \emph{Qubit implementation.} Truncate each local quadrature to occupations $0,\ldots,s$ and encode each submode in $\lceil\log_2(s+1)\rceil$ qubits. Prepare the distributed Fock input, apply the six finite shear circuits, and measure the submode occupations. For each original output mode, sum the occupations of all $2K$ submodes belonging to its two output ports. Report a failure symbol $\perp$ if the total differs from $n$; this outcome is included in the unconditional output distribution.
\end{enumerate}
Here, $s$ is the local occupation cutoff and $K$ is the number of submodes per optical mode. Their values, together with the gate precision, are chosen in the subsequent error analysis to achieve the desired error $\epsilon$ and resource bounds.

\begin{figure}[tb]
\centering
\includegraphics[width=0.48\textwidth]{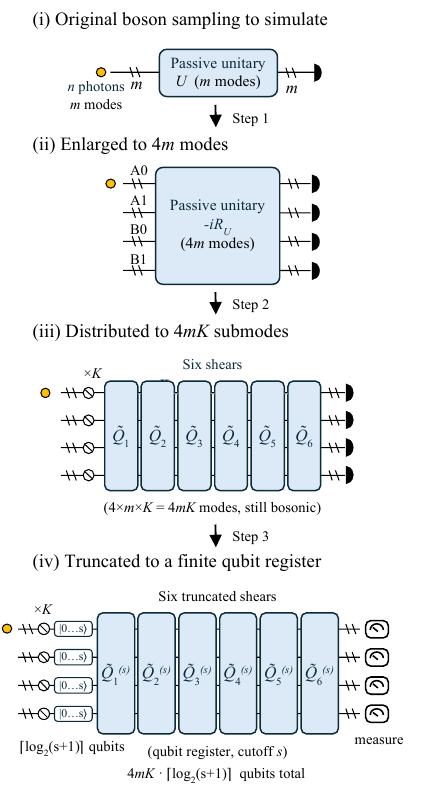}
\caption{Construction of the boson-sampling simulator. (i) The original $m$-mode problem with $n$ single-photon inputs. (ii) Enlargement to $4m$ modes, with the fixed input encoder implicit before $-iR_U$. (iii) Distribution over $K$ submodes per mode and the six corresponding shears. (iv) Local truncation to occupations $0,\ldots,s$ and qubit implementation. Summing the measured occupations over the submodes of each output pair gives the simulated sample. The displayed qubit count includes only the submode registers; input and auxiliary registers are omitted.}
\label{fig:construction}
\end{figure}

Sections~\ref{sec:gaussian}--\ref{sec:submodes} establish these three steps; Section~\ref{sec:error} bounds their approximation error and proves the resource estimates.

\section{Step 1: Enlargement and shear decomposition}\label{sec:gaussian}
We first enlarge the original $m$-mode interferometer to a passive optical circuit on $4m$ modes with vacuum auxiliary inputs. Photon counting followed by summing fixed pairs of output occupations reproduces the original distribution. This circuit consists of a balanced input encoder followed by a passive transformation with a real symmetric mode matrix squaring to the identity; each is decomposed into three quadratic shears.

Label the four groups of $m$ modes by $A0,A1,B0,B1$, in this order, and index the modes within each group by $j=1,\ldots,m$. The original input occupies $A0$, while $A1,B0,B1$ initially contain vacuum. The output will be read from the pairs $(B0_i,B1_i)$ by adding the two photon counts for each $i$.

\begin{proposition}[Exact six-shear dilation]\label{prop:six}
For every $U\in U(m)$, there is a circuit of six quadratic shears on the $4m$ modes above whose product is passive. For the specified input, $A0,A1$ return to vacuum, while summing the output counts in each pair $(B0_i,B1_i)$ exactly reproduces the output distribution of $U$. Each shear coefficient matrix is real symmetric, zero-diagonal, and has norm at most one; the input ports and output pairs are fixed independently of $U$.
\end{proposition}

The construction has two parts. We first build an equivalent passive sampling circuit using a fixed encoder followed by a transformation with real symmetric mode matrix $R_U$ satisfying $R_U^2=I$. We then express the encoder and the transformation defined by $R_U$, up to a phase that does not affect photon counting, as two groups of three shears. The zero diagonal of the shear coefficients will make the local Hamiltonian compression exact, while bounded intermediate squeezing will control the approximation error.

\subsection{An equivalent interferometer on four times as many modes}
We reach the $4m$-mode circuit through two successive enlargements. First combine the original input group $A0$ with the vacuum group $A1$. This gives a $2m$-mode representation in which the action of the complex matrix $U$ is encoded by a real orthogonal matrix. We then include the vacuum groups $B0,B1$ to turn that real orthogonal matrix into a symmetric involution suitable for the shear decomposition.

For the first enlargement, write $U=X+iY$, with $X,Y$ real, and define the real orthogonal matrix
\begin{align}
O_U\coloneqq\begin{pmatrix}X&-Y\\Y&X\end{pmatrix}.
\label{eq:realembedding}
\end{align}
To connect this real representation to the physical input arrangement, introduce the balanced encoder on $A0,A1$,
\begin{align}
C\coloneqq\frac1{\sqrt2}\begin{pmatrix}I_m&-iI_m\\-iI_m&I_m\end{pmatrix}.
\label{eq:balancedencoder}
\end{align}
Since $A1$ begins in vacuum, $U\oplus U^*$, where $U^*$ denotes entrywise complex conjugation, reproduces the action of the original interferometer on $A0$. Applying $C$ afterward preserves the total photon number within each output pair. Thus, reporting $n_i=N_{A0,i}+N_{A1,i}$ gives exactly $P_U$. We can now rewrite this circuit using the identity
\begin{align}
C(U\oplus U^*)=O_UC,
\label{eq:intertwine}
\end{align}
which gives an equivalent realization: apply the fixed encoder $C$ to the input, followed by the real orthogonal transformation $O_U$. Thus, the first enlargement realizes the original sampling problem using the real orthogonal transformation $O_U$; the second will give the symmetric involution needed for the shear decomposition.

For the second enlargement, append the vacuum groups $B0,B1$ and extend the input encoder to $C\oplus I_{2m}$. With the mode order $A0,A1,B0,B1$, define

\begin{align}
R_U\coloneqq\begin{pmatrix}0&O_U^{\T}\\O_U&0\end{pmatrix}.
\label{eq:reflection}
\end{align}
In the first enlargement, $O_U$ is the unitary matrix describing a passive interferometer on the modes; here, $R_U$ plays the same role on the enlarged system. The block form of $R_U$ applies the same transformation $O_U$ to the occupied $A$ groups while transferring the output to $B0,B1$. The $A$ groups return to vacuum, and summing the photon counts in each pair $(B0_i,B1_i)$ again gives $P_U$. The complementary block $O_U^{\T}$ makes $R_U$ real orthogonal and symmetric. The involution property $R_U^2=I_{4m}$ gives the fixed, bounded shear strengths used below, while the zero diagonal will be needed later in the local-truncation argument.

\subsection{Shear factorization}
Throughout this subsection, $\hat q,\hat p$ collect all $4m$ modes in the order $(A0,A1,B0,B1)$, with indices $1,\ldots,m$ within each group; the single indices $j,k=1,\ldots,4m$ enumerate the modes in this same order.
We now express the encoder and the interferometer defined by $R_U$ as six Gaussian shears on the same $4m$ modes, preserving their photon-counting distribution. To connect the mode transformation to the operator being decomposed, note that the real symmetry of $R_U$ makes the following number-preserving quadratic generator Hermitian. Define the passive Fock-space unitary
\begin{align}
\hat V_U\coloneqq\exp\!\left[-\frac{i\pi}{2}\sum_{j,k=1}^{4m}\hat a_j^\dagger(R_U)_{jk}\hat a_k\right].
\label{eq:fockreflection}
\end{align}
For the state transformation $\ket\psi\mapsto\hat V_U\ket\psi$, the corresponding mode transformation is
\begin{align}
\hat V_U \hat a_j^\dagger\hat V_U^\dagger
=\sum_{k=1}^{4m}\bigl[e^{-i\pi R_U/2}\bigr]_{kj}\hat a_k^\dagger.
\label{eq:reflectionbch}
\end{align}
The $4m\times 4m$ matrix exponential thus appears as the coefficients mixing the creation operators. Using $R_U^2=I_{4m}$ then evaluates it as
\begin{align}
e^{-i\pi R_U/2}=-iR_U.
\label{eq:reflectionevolution}
\end{align}
Consequently, $\hat V_U \hat a_j^\dagger\hat V_U^\dagger=-i\sum_{k=1}^{4m}(R_U)_{kj}\hat a_k^\dagger$ for $1\le j\le 4m$. This is the mode transformation generated by $R_U$ from the previous subsection with a common factor $-i$ per photon. On an $n$-photon input the factor is the global phase $(-i)^n$; hence, the distribution of photon counts summed over each output pair is unchanged.

We now decompose this Fock space unitary. Since $R_U$ is real symmetric and has trace zero, Eq.~\eqref{eq:fockreflection} can be written in quadratures as
\begin{align}
\hat V_U=\exp\!\left[-\frac{i\pi}{4}\left(\hat q^{\T}R_U\hat q+\hat p^{\T}R_U\hat p\right)\right].
\label{eq:reflectionquadratures}
\end{align}
Diagonalizing $R_U$ reduces this evolution to independent rotations in the position--momentum planes of its eigenmodes. Since the involution property $R_U^2=I_{4m}$ restricts the eigenvalues to $\pm1$, applying the standard factorization of these planar rotations into three shears~\cite{Paeth1986} and returning to the original mode basis gives the desired decomposition,
\begin{align}
\hat V_U=\hat Q_q(R_U)\hat Q_p(R_U)\hat Q_q(R_U),
\label{eq:centralthreeshears}
\end{align}
where, for real symmetric matrices $A,B$,
\begin{align}
\hat Q_q(A)\coloneqq e^{-\frac i2\hat q^{\T}A\hat q},\qquad \hat Q_p(B)\coloneqq e^{-\frac i2\hat p^{\T}B\hat p}.
\label{eq:shearoperators}
\end{align}
The diagonalization is used only to derive this identity and is not implemented by the circuit. Appendix~\ref{app:shears} gives the rotation identity, its application to the eigenmodes of $R_U$, and the verification of the overall phase.

The balanced input encoder similarly consists of rotations by $\pm\pi/4$ on its two eigenmodes per pair. Its three shear coefficients, in the fixed mode order, are
\begin{align}
B_{\mathrm{enc}}\coloneqq\left[\frac1{\sqrt2}\begin{pmatrix}0&I_m\\I_m&0\end{pmatrix}\right]\oplus0_{2m},\quad
A_{\mathrm{enc}}\coloneqq(2-\sqrt2)B_{\mathrm{enc}}.
\label{eq:encodershears}
\end{align}
The product $\hat Q_q(A_{\mathrm{enc}})\hat Q_p(B_{\mathrm{enc}})\hat Q_q(A_{\mathrm{enc}})$ implements the encoder $C\oplus I_{2m}$. These coefficients follow from the same rotation identity at angles $\pm\pi/4$, as derived in Appendix~\ref{app:shears}.

\begin{proof}[Proof of Proposition~\ref{prop:six}]
Apply the three encoder shears from Eq.~\eqref{eq:encodershears}, followed by the three shears implementing $\hat V_U$ in Eq.~\eqref{eq:centralthreeshears}. All six coefficient matrices have zero diagonal and norm at most one. The resulting mode transformation is that of the enlarged sampling circuit, with a common factor $-i$ per photon. For an $n$-photon input this is the global phase $(-i)^n$; hence, the pair-sum readout gives exactly $P_U$.
\end{proof}
Step 1 specifies the input and six-shear circuit on $4m$ modes. We next distribute these modes over more optical modes while preserving the sampling distribution.

\section{Step 2: Redistribution over submodes}\label{sec:redistribution}
We now enlarge the six-shear circuit from $4m$ modes to $4mK$ submodes in order to reduce the occupation carried by any one local mode. Each original mode is spread coherently over $K$ submodes, suppressing large occupation of an individual submode. Step 3 will use this dilution to truncate every submode at a fixed local cutoff while controlling the resulting error. The enlargement is exact: as in the first enlargement, summing the counts within each output block preserves the sampling distribution.

Start with the $4m$-mode input from Step 1, with photons in the occupied $A0$ modes and vacuum in the other modes. Since the input encoder is already included in the six shears, no shear has yet been applied. Attach $K-1$ vacuum submodes to each of the $4m$ modes. Write the original mode in each block as $\hat a_{i1}$ and the added modes as $\hat a_{i2},\ldots,\hat a_{iK}$, where $1\le i\le4m$. Let $\hat{\mathcal D}_K$ be the product of passive $K$-port splitters, one for each original mode, chosen so that
\begin{align}
\hat a_{i1}^\dagger\mapsto
\frac1{\sqrt K}\sum_{\mu=1}^K \hat a_{i\mu}^\dagger.
\end{align}
Thus, the occupation of mode $i$ is spread uniformly across its $K$ submodes. The operator $\hat{\mathcal D}_K$ is unitary on the full $4mK$-mode Fock space, and the distributed input is $\hat{\mathcal D}_K(\ket\psi\otimes\ket0_{\rm aux})$, where $\ket\psi$ is the original $4m$-mode input.

Applying $\hat{\mathcal D}_K$ only to the input would dilute the state but would not reproduce the original circuit: the subsequent shears must act on the collective modes rather than on a single representative submode. We therefore make the same change of modes in every shear. Let $\hat Q_1,\ldots,\hat Q_6$ denote the original shears in chronological order, including the three input-encoder shears, and define
\begin{align}
\hat{\widetilde Q}_\ell\coloneqq\hat{\mathcal D}_K(\hat Q_\ell\otimes \hat I_{\rm aux})\hat{\mathcal D}_K^\dagger,
\qquad 1\le\ell\le6.
\label{eq:dilutedshear}
\end{align}
Each $\hat{\widetilde Q}_\ell$ is again a shear: replace each original quadrature by the sum of its $K$ submode quadratures divided by $\sqrt K$. For example,
\begin{align}
\hat{\widetilde Q}_q(A)\coloneqq\exp\!\left[-\frac{i}{2K}\sum_{i,j=1}^{4m}A_{ij}
\sum_{\mu,\nu=1}^K\hat q_{i\mu}\hat q_{j\nu}\right].
\label{eq:dilutedq}
\end{align}
Its coupling between submodes $(i,\mu)$ and $(j,\nu)$ is $A_{ij}/K$. It applies the original shear to the collective modes and leaves the orthogonal mode combinations unchanged. The same replacement gives the enlarged $p$-shear.

The enlarged circuit first applies $\hat{\mathcal D}_K$ to the input with auxiliary vacuum, and then applies $\hat{\widetilde Q}_1,\ldots,\hat{\widetilde Q}_6$. Multiplying Eq.~\eqref{eq:dilutedshear} cancels each adjacent $\hat{\mathcal D}_K^\dagger\hat{\mathcal D}_K$, giving
\begin{align}
&\hat{\widetilde Q}_6\cdots\hat{\widetilde Q}_1
\hat{\mathcal D}_K(\ket\psi\otimes\ket0_{\rm aux})\nonumber\\
&\qquad=\hat{\mathcal D}_K\bigl[(\hat Q_6\cdots \hat Q_1\ket\psi)\otimes\ket0_{\rm aux}\bigr].
\label{eq:dilutedcircuit}
\end{align}
The state in square brackets on the right, before the leading $\hat{\mathcal D}_K$, is exactly the output of the original $4m$-mode six-shear circuit, embedded with auxiliary vacuum. Hence, the enlarged circuit differs from the original one only by $\hat{\mathcal D}_K$ at the output. This remaining splitter mixes submodes only within each block and preserves the block's total photon number. We can therefore measure all submode occupations and add them within each block to recover exactly the original $4m$-mode photon-counting distribution. Adding in turn the two blocks associated with each $(B0_i,B1_i)$ output pair, as in Step 1, gives the original boson-sampling distribution. Thus, no inverse splitters are needed at the output.

We have now specified an exact bosonic sampler on $4mK$ modes: prepare the distributed input, apply six enlarged shears, measure, and add the counts. Each submode still has its full Fock space. Increasing $K$ spreads the occupation over more submodes; the next step uses this to construct a finite qubit circuit.

\section{Step 3: Qubit implementation}\label{sec:submodes}
We now turn the $4mK$-mode sampler from Step 2 into a qubit circuit. We truncate each submode's quadratures to occupations $0,\ldots,s$ and exponentiate the resulting finite shear Hamiltonians. Each retained Fock basis is encoded in qubits. Local basis changes and coherent counting implement the six finite shears in logarithmic depth. We conclude with explicit input preparation and measurement postprocessing, completing the sampler; Section~\ref{sec:error} then chooses $s$, $K$, and the gate precision.

\subsection{Local truncation and qubit encoding}
Introduce an occupation cutoff $s\ge1$, whose value will be chosen in Section~\ref{sec:resources} together with the number of submodes $K$. In each submode retain the Fock states $\ket0,\ldots,\ket s$ and define the truncated annihilation and quadrature operators
\begin{align}
\hat{\bar c}\coloneqq\sum_{r=1}^s\sqrt r\ket{r-1}\bra r,~~
\hat{\bar q}\coloneqq\frac{\hat{\bar c}+\hat{\bar c}^\dagger}{\sqrt2},~~
\hat{\bar p}\coloneqq\frac{\hat{\bar c}-\hat{\bar c}^\dagger}{i\sqrt2}.
\label{eq:submode}
\end{align}

The exact $q$-shear on the submodes has Hamiltonian
\begin{align}
\hat H(A)\coloneqq\frac1{2K}\sum_{i,j=1}^{4m} A_{ij}\sum_{\mu,\nu=1}^K \hat q_{i\mu}\hat q_{j\nu}.
\end{align}
Substituting $\hat{\bar q}$ for each local quadrature gives the finite Hamiltonian
\begin{align}
\hat{\bar H}(A)\coloneqq\frac1{2K}\sum_{i,j=1}^{4m} A_{ij}
\left(\sum_{\mu=1}^K \hat{\bar q}_{i\mu}\right)
\left(\sum_{\nu=1}^K \hat{\bar q}_{j\nu}\right).
\label{eq:finiteH}
\end{align}

Exponentiating the finite Hamiltonian defines the truncated shear $\hat{\bar Q}_q(A)\coloneqq e^{-i\hat{\bar H}(A)}$. Replacing $\hat{\bar q}$ by $\hat{\bar p}$ gives $\hat{\bar Q}_p(A)$. These are unitary operators on the tensor product of the $4mK$ truncated local Fock spaces.

We encode the truncated Fock basis of each submode in a qubit register. The register stores its occupation number through $\ket r\mapsto\ket{\operatorname{bin}(r)}$, $0\le r\le s$, using $b_s\coloneqq\lceil\log_2(s+1)\rceil$ qubits per submode, giving $4mKb_s$ data qubits in total. The cutoff $s$ determines the size of the local operators, while $K$ controls the dilution. Their eventual choices balance the local register size against the truncation error.

When $s+1$ is a power of two, all computational basis states of the submode register are used. For other $s$, extend $\hat{\bar q},\hat{\bar p}$ by zero on unused labels and extend the local basis changes unitarily. The finite dynamics preserves the retained submode subspace. At readout, every binary label is interpreted as its nonnegative integer occupation, including labels above $s$. This convention identifies the entire qubit register with the corresponding Fock levels and allows the same photon-counting rule to be applied to the finite and ideal states.

\subsection{Parallel implementation of a truncated shear}
Our goal is to implement the truncated shear coherently on an arbitrary state of the submode registers. For the $q$-shear, the required map is
\begin{align}
\ket A_P\ket\psi_D\ket0_W
\longmapsto
\ket A_P\bigl(e^{-i\hat{\bar H}(A)}\ket\psi_D\bigr)\ket0_W.
\nonumber
\end{align}
Here, $P$ holds the binary coefficients, $D$ contains the $4mK$ submode registers, and $W$ contains temporary counting and fanout registers. The coefficients remain unchanged and all temporary registers return to zero. The map must preserve superpositions, including those produced by preceding shears. The $p$-shear has the same requirement with $\hat{\bar q}$ replaced by $\hat{\bar p}$ in the Hamiltonian.

To realize this map, we analyze the Hamiltonian in Eq.~\eqref{eq:finiteH}. A local basis change makes it diagonal; hence, its exponential can be implemented by applying the corresponding phase to each basis string. We compute the counts that determine this phase, apply the phase coherently, and uncompute the counts before returning to the occupation basis. The following proposition states the depth and qubit cost of this construction.

\begin{proposition}[Qubit circuit for a truncated shear]\label{prop:shearresource}
Fix an integer $s\ge1$. The constants hidden in $O(\cdot)$ in this subsection may depend on this fixed cutoff. Suppose each coefficient $A_{ij}$ is supplied as a two's-complement fixed-point number of at most $b_A$ bits. For $K\ge2$, the encoded unitary $\hat{\bar Q}_q(A)$ or $\hat{\bar Q}_p(A)$ can be implemented exactly in the model with continuous gates with
\begin{align}
D_s&=O(\log(4mK)+\log(b_A+1)),\nonumber\\
W_s&=O(4mK+(4m)^2b_A\log^2K).
\label{eq:shearresource}
\end{align}
All temporary counting and fanout registers return to zero for every submode input. Six shears preserve the same asymptotic bounds. The value of $b_A$ required for the final accuracy is chosen in Section~\ref{sec:resources}.
\end{proposition}

In our sampler, the input consists of the matrix-entry bits of $U$, rather than a separate list of shear coefficients. By Step 1, each nonconstant coefficient is a signed real or imaginary part of an entry of $U$; the encoder coefficients are fixed constants. Thus, the proposition uses the input bits of $U$ with the prescribed indexing and signs, together with fixed approximations to the encoder constants. No input-dependent matrix decomposition is required.

The Hamiltonian is quadratic in local quadratures acting on separate registers. Its diagonal form therefore follows from the eigendecomposition of a single local truncated position operator,
\begin{align}
\hat{\bar q}=\hat{\mathcal W}_s\diag(\lambda_0,\ldots,\lambda_s)\hat{\mathcal W}_s^\dagger.
\end{align}
For fixed $s$, this matrix has constant dimension and is independent of $U$. Its eigenvalues and eigenbasis matrix are computed classically to the precision needed to synthesize the circuit over a finite gate set, as described in Appendix~\ref{sec:uniformity}. The required precision, and hence the classical preparation cost, can increase with system size. The quantum circuit applies the resulting gates for the basis change; it does not run an eigendecomposition algorithm. Only this local eigendecomposition is prepared classically; the spectrum of the full Hamiltonian is never tabulated. The local eigenvalues enter the fixed parameters of the phase gates.

Write $\ket{v_r}\coloneqq\hat{\mathcal W}_s\ket r$. Here, $\ket r$ always denotes the computational basis state of a submode register. Before the basis change it encodes occupation $r$; after applying $\hat{\mathcal W}_s^\dagger$, the same label $r$ identifies the quadrature eigenvector $\ket{v_r}$ and hence the eigenvalue $\lambda_r$. Accordingly, the counts $h_{i,r}$ below count eigenvalue labels rather than occupations. Applying $\hat{\mathcal W}_s^\dagger$ to each submode maps $\ket{v_r}$ to $\ket r$, preserving superpositions. Let $\hat{\mathcal V}_s\coloneqq\hat{\mathcal W}_s^{\otimes 4mK}$ and write $\ket{\bm r}\coloneqq\bigotimes_{i,\mu}\ket{r_{i\mu}}$. For each label string, let $h_{i,r}\coloneqq\#\{\mu:r_{i\mu}=r\}$. The shear in these coordinates acts as
\begin{align}
&\hat{\mathcal V}_s^\dagger\hat{\bar Q}_q(A)\hat{\mathcal V}_s\ket{\bm r}
=e^{-iE_A(\bm h(\bm r))}\ket{\bm r},\\
&E_A(\bm h)\coloneqq\frac1{2K}\sum_{i,j=1}^{4m}A_{ij}
\left(\sum_{r=0}^s\lambda_rh_{i,r}\right)
\left(\sum_{t=0}^s\lambda_th_{j,t}\right).
\label{eq:histphase}
\end{align}
Here, $E_A(\bm h)$ is the Hamiltonian eigenvalue, and $e^{-iE_A(\bm h)}$ is the shear eigenvalue for the corresponding joint eigenbasis state. It depends on the label string through its counts; hence, it is not a common global phase.

The full circuit first applies $\hat{\mathcal V}_s^\dagger$ to change from the occupation basis to the quadrature-eigenvalue labels. It then computes the counts and applies the diagonal phase. Writing $H$ for the histogram subregister of $W$, these middle steps are
\begin{align}
\ket A_P\ket{\bm r}_D\ket0_H
&\longmapsto
\ket A_P\ket{\bm r}_D\ket{\bm h(\bm r)}_H\nonumber\\
&\longmapsto
e^{-iE_A(\bm h(\bm r))}
\ket A_P\ket{\bm r}_D\ket{\bm h(\bm r)}_H\nonumber\\
&\longmapsto
e^{-iE_A(\bm h(\bm r))}
\ket A_P\ket{\bm r}_D\ket0_H.
\end{align}
The three displayed maps act coherently: the first computes the histogram, the second applies the phase, and the third uncomputes the histogram. Finally, $\hat{\mathcal V}_s$ returns the data registers to the occupation basis. For $\hat{\bar p}$, the corresponding matrix for the local basis change is $\diag(i^r)\hat{\mathcal W}_s$, with $r=0,\ldots,s$.

For fixed $s$, all counts $h_{i,r}$ can be computed coherently and in parallel using quantum carry-save arithmetic~\cite{Gossett} followed by a ripple-carry adder~\cite{Cuccaro}, in depth $O(\log K)$, using $O(4mK)$ auxiliary qubits. Expanding the counts and the $b_A$-bit coefficients in binary expresses $e^{-iE_A(\bm h)}$ as a product of $O((4m)^2b_A\log^2K)$ constant-arity controlled phase gates. The parallelization of diagonal gates of Moore and Nilsson~\cite[Proposition 5]{MNcodes} uses CNOT trees to provide disjoint copies of the control bits, allowing all these gates to be applied in parallel. The fanout, phase application, and reversal have depth $O(\log(4mK)+\log(b_A+1))$. Reversing the fanout and counting circuits returns all auxiliary registers to zero.

Together with the constant-depth local basis changes, these constructions give the resource scaling stated in Proposition~\ref{prop:shearresource}. Appendix~\ref{sec:uniformity} provides the detailed proof, including the ancilla count, and establishes uniform circuit generation.

In the precision regime chosen in Section~\ref{sec:error}, $\log(b_A+1)=O(\log(4mK))$; hence, Proposition~\ref{prop:shearresource} gives $D_s=O(\log(4mK))$. This construction is exact in the model with continuous gates. Section~\ref{sec:resources} approximates the local basis changes and phase rotations over Clifford+$T$ at the precision required for the final sampling guarantee.
\subsection{Input preparation and photon number readout}\label{sec:readout}
The circuit starts in the qubit encoding of the split Fock input and ends with occupation measurements. Both the input preparation and the processing of the output have logarithmic depth.

The input required by Step 2 is the qubit encoding of $\hat{\mathcal D}_K(\ket\psi\otimes\ket0_{\rm aux})$. Since each occupied input mode contains one photon and every other input mode is in vacuum, it suffices to implement the splitter on the vacuum and one-photon sectors of each block. Vacuum is represented by all zeros and is fixed, while $\hat{\mathcal D}_K$ maps a one-photon input to
\begin{align}
\ket1\mapsto \frac{1}{\sqrt{K}}\sum_{\mu=1}^K\ket{0\cdots1_\mu\cdots0}.
\label{eq:onephotondilution}
\end{align}
We choose $K$ to be a power of two in Section~\ref{sec:resources}; rounding upward changes it by less than a factor of two. Starting with one excitation in the first submode, use a two-qubit splitter that fixes $\ket{00}$ and maps $\ket{10}$ to $(\ket{10}+\ket{01})/\sqrt2$. One exact implementation is CNOT from the second bit to the first, controlled Hadamard from the first to the second, and the same CNOT again. On $\ket{10}$, the three stages give $\ket{10}\mapsto\ket{10}\mapsto(\ket{10}+\ket{11})/\sqrt2\mapsto(\ket{10}+\ket{01})/\sqrt2$. Controlled Hadamard has an exact constant-size Clifford+$T$ decomposition. Applying these splitters in parallel along a balanced binary tree doubles the number of populated branches in every layer, so $\log_2K$ layers prepare Eq.~\eqref{eq:onephotondilution}. Only the least significant occupation bit of each submode register is used, and the higher bits remain zero.

When $n$ is supplied in binary, reversibly compute $f_j\coloneqq[j\le n]$ for $j=1,\ldots,m$ and write $f_j$ directly into the least significant bit of the first submode of $A0$ block $j$, uncomputing the comparison workspace. All $m$ comparisons run in parallel in $O(\log m)$ depth. Applying the splitting trees to all blocks in parallel then prepares the required input in depth $O(\log(mK))$. The prepared state has local occupation at most one and therefore lies in the truncated subspace.

Combining the input preparation, the truncated shears, and the readout gives the complete qubit sampler for fixed $s$ and $K$:
\begin{enumerate}
\item Prepare one uniform excitation over the $K$ submodes of every occupied $A0$ block and prepare all other submodes in vacuum.
\item Apply the six truncated shears in the order specified by Proposition~\ref{prop:six}, using the local basis changes and parallel phases above.
\item Measure every submode register. If $n_{B0,i\mu}$ and $n_{B1,i\mu}$ are the measured occupations in the two blocks for output pair $i$, report $n_i\coloneqq\sum_{\mu=1}^K\bigl(n_{B0,i\mu}+n_{B1,i\mu}\bigr)$, $1\le i\le m$.
Discard the $A0,A1$ measurement records. If the reported total $\sum_{i=1}^m n_i$ differs from $n$, output the failure symbol $\perp$.
\end{enumerate}
On the ideal output, the failure symbol never occurs and the reported pattern has distribution $P_U$ by Step 2. The state approximation bound below therefore controls the total-variation distance including the failure outcome, without postselection. If only valid patterns are desired, mapping $\perp$ to any fixed valid pattern cannot increase this distance.

All sums and the test of total photon number can also be computed coherently before measurement. Apply the same carry-save constructions and final adders~\cite{Gossett,Cuccaro} to the $2K$ occupation registers of fixed width in each output pair, and compare the reported total with the specified photon number $n$. Parallel summation, comparison, and computation of the failure flag have depth $O(\log(mK))$ and use $O(mK)$ auxiliary qubits. Thus, either classical postprocessing of the measured occupations or coherent readout gives the same distribution within our resource bounds. The following section chooses $K$ and the gate precision to bound its error by $\epsilon$.

\section{Approximation error and resource bounds}\label{sec:error}
This section proves that the finite qubit sampler from Step 3 approximates the exact bosonic sampler from Step 2. There are three errors to control: truncating each submode, using finite-precision coefficients, and synthesizing the elementary rotations over Clifford+$T$. We first bound the truncation error by following the ideal state through the six shears. We then choose $s$, $K$, and the two finite precisions so that the sum of the three errors is at most $\epsilon$.

\subsection{Dynamical error from truncated submodes}
To prove Theorem~\ref{thm:approx}, we must show that the local truncation in Step 3 preserves the target distribution to the required accuracy. In particular, increasing the number $K$ of submodes must suppress the error while the cutoff $s$ remains fixed. Small leakage at the output alone is not sufficient, because truncation changes the evolution throughout the six shears. The following theorem bounds the resulting error in the full output state.

\begin{theorem}[Local truncation error]\label{thm:truncationerror}
Let $\ket{\Psi_{\mathrm{id}}(6)}$ and $\ket{\Psi_{\mathrm{tr}}(6)}$ denote the final states of the ideal and locally truncated six-shear circuits, respectively, starting from the specified single-photon Fock input after splitting into submodes. Then
\begin{align}
\norm{\ket{\Psi_{\mathrm{tr}}(6)}-\ket{\Psi_{\mathrm{id}}(6)}}
\le C_s\frac{(n+4m+2s+8)^{(s+3)/2}}{K^{s/2}},
\label{eq:truncationbound}
\end{align}
where $C_s\coloneqq(1+6\sqrt3)2^{s+3}/\sqrt{(s+1)!}$.
\end{theorem}

Here the truncated state is identified with its natural embedding in the full $4mK$-mode Fock space, and the comparison is made before coefficient rounding and elementary-gate approximation. For fixed $s$, the bound decreases as $K^{-s/2}$; hence, polynomially many submodes suffice for inverse-polynomial truncation error. The proof has three ingredients. First, let $\hat P$ project onto occupations $0,\ldots,s$ in every submode and put $\hat P^\perp\coloneqq\hat I-\hat P$. Because every shear coefficient matrix has zero diagonal, its quadratic terms act on distinct submodes. The local truncation therefore gives exactly the compressed Hamiltonian,
\begin{align}
\hat{\bar H}(A)=\hat P\hat H(A)\hat P\big|_{\im\hat P}.
\nonumber
\end{align}
Second, parametrize each of the six shears by one unit of time, with $\hat H_K(t)$ the ideal Hamiltonian on the submodes and $\ket{\Psi_{\mathrm{id}}(t)}$ its trajectory. Comparing the truncated evolution with the full ideal evolution, a Duhamel argument gives
\begin{align}
&\norm{\ket{\Psi_{\mathrm{tr}}(6)}-\ket{\Psi_{\mathrm{id}}(6)}}\nonumber\\
&\quad\le\norm{\hat P^\perp\ket{\Psi_{\mathrm{id}}(6)}}+
\int_0^6\norm{\hat P\hat H_K(t)\hat P^\perp\ket{\Psi_{\mathrm{id}}(t)}}\,dt.
\label{eq:comparisonoverview}
\end{align}
Thus, both the final leakage and its dynamical effect must be controlled. The Hamiltonian grows at most linearly with total photon number, reducing both terms to a weighted leakage estimate along the ideal trajectory.

Finally, uniform splitting suppresses this weighted leakage by $K^{-s/2}$, with a prefactor determined by photon number moments on the original $4m$ modes. Although individual shears can temporarily create photons, the quadratures are amplified by at most a constant factor throughout the six-shear evolution. This bounds the required moments and yields Theorem~\ref{thm:truncationerror}. Appendix~\ref{app:error} gives the precise embedding, the local compression identity, and the complete proof.

\subsection{Parameter choices, circuit resources, and sampling guarantee}\label{sec:resources}
We now choose the construction parameters and combine the three errors with the readout from Step 3. Fix a constant integer cutoff $s\ge1$, independent of $n,m,\epsilon$. Accuracy is obtained by increasing $K$, while the local register dimension stays fixed. The choice $s=1$ suffices for the main theorem; for the corollary giving width arbitrarily close to quadratic, we choose a larger constant $s$ depending on the requested accuracy and width exponents. Throughout this subsection, constants hidden in $O(\cdot)$ may depend on this fixed cutoff. We first choose $K$, then the coefficient precision, and finally the accuracy of each synthesized rotation.

For total error $0<\epsilon<1$, allocate $\epsilon/3$ each to the truncation error, coefficient precision, and gate synthesis. Theorem~\ref{thm:truncationerror} shows that the first contribution is at most $\epsilon/3$ if $K$ is a power of two satisfying
\begin{align}
K\ge
\left(\frac{3C_s}{\epsilon}\right)^{2/s}
(n+4m+2s+8)^{(s+3)/s}.
\label{eq:Kchoice}
\end{align}
For a single circuit layout that accepts $n$ as an input, choose the smallest power of two satisfying Eq.~\eqref{eq:Kchoice} with $n$ replaced by $m$. This choice is independent of the input bits and, since $n\le m$, gives
\begin{align}
K=O\!\left(m^{1+3/s}\epsilon^{-2/s}\right).
\end{align}

Let $\widetilde A$ be the coefficient matrix obtained from the finite-precision entries supplied for $U$ or from approximating the fixed encoder coefficients, with the same symmetry and zero diagonal as the ideal matrix $A$, and suppose
\begin{align}
\abs{A_{ij}-\widetilde A_{ij}}\le\eta_A
\qquad\text{for all }i,j.
\label{eq:coefficienterror}
\end{align}
Since $\norm{\hat{\bar q}}\le\sqrt{2s}$, each normalized collective quadrature $K^{-1/2}\sum_{\mu=1}^K\hat{\bar q}_{i\mu}$ in Eq.~\eqref{eq:finiteH} has norm at most $\sqrt{2sK}$. Submultiplicativity bounds the contribution of one coefficient error by $sK\eta_A$. The triangle inequality over the $(4m)^2$ entries gives
\begin{align}
\norm{\hat{\bar H}(A)-\hat{\bar H}(\widetilde A)}\le (4m)^2sK\eta_A.
\label{eq:Hcoefficienterror}
\end{align}
For Hermitian matrices $H$ and $\widetilde H$, the standard bound $\|e^{-iH}-e^{-i\widetilde H}\|\le\|H-\widetilde H\|$ and Eq.~\eqref{eq:Hcoefficienterror} bound the error of one shear by $(4m)^2sK\eta_A$. Telescoping the products of the six shears therefore gives a total coefficient error of at most $6(4m)^2sK\eta_A$. With a fixed number of sign and integer bits, a $b_A$-bit fixed-point coefficient has rounding error $O(2^{-b_A})$. It is therefore sufficient to choose
\begin{align}
\eta_A&\le\frac{\epsilon}{18(4m)^2sK},\\
b_A&=\left\lceil\log_2\frac{18(4m)^2sK}{\epsilon}\right\rceil+O(1)=O\!\left(\log\frac{mK}{\epsilon}\right).
\label{eq:entryprecision}
\end{align}
Here, $b_A$ is the coefficient bit length introduced in Proposition~\ref{prop:shearresource}. The real and imaginary parts of $U$ are requested to this precision, while the fixed encoder coefficients are approximated to the same tolerance. Signed copies of an input entry reuse its bits with the corresponding fixed signs. These finite-precision coefficients directly specify the phase gates, and the bound compares their evolution with that defined by the ideal $U$.

We next replace the continuous one-qubit rotations by Clifford+$T$ circuits. First decompose each local basis change of fixed dimension and controlled phase gate into one-qubit gates and CNOTs~\cite{Barenco}. Since $s$ is fixed, this introduces only a constant number of elementary stages per primitive. We use Clifford+$T$ synthesis with $O(\log(1/\delta))$ gates and at most two initially zero ancillas to approximate each single-qubit unitary to error $\delta$, uniformly over data inputs and including the ancillary output~\cite{KMM}. For general approximation over a finite gate set and ancilla-free $z$-rotation synthesis, see Refs.~\cite{DawsonNielsen,Ross}.

Let $G_{\rm rot}$ count the occurrences of one-qubit rotations requiring approximation in this decomposition. Across the six shears, the transformations $\hat{\mathcal V}_s^\dagger$ and $\hat{\mathcal V}_s$ surrounding each diagonal $q$-shear phase, together with the analogous $p$-basis transformations, act locally on all $4mK$ submode registers and contribute $O(4mK)$ rotations. Expanding the diagonal phases over the coefficient bits and pairs of eigenvalue-count bits contributes $O((4m)^2b_A\log^2K)$ more. Thus,
\begin{align}
G_{\rm rot}=O(4mK+(4m)^2b_A\log^2K).
\end{align}
Appendix~\ref{sec:uniformity} gives this count explicitly in Eq.~\eqref{eq:rotationcountappendix}. Input splitting, histogram counting, fanout, reversible arithmetic, and all their reversals use exact Clifford+$T$ gates and require no precision allocation.

Give each rotation its own initially zero synthesis ancillas, which are not reused by later rotations. The synthesis guarantee compares the joint data--ancilla output with the ideal rotation and zero ancillas, uniformly over the data state, including entanglement with other registers~\cite{KMM}. Replacing the rotations one at a time therefore changes the full output state by at most $G_{\rm rot}\delta$. To allocate at most $\epsilon/3$ to synthesis, choose
\begin{align}
\delta\le\frac{\epsilon}{3G_{\rm rot}}.
\end{align}
Each rotation then has a Clifford+$T$ implementation of length $O(\log(G_{\rm rot}/\epsilon))$.

This length contributes to depth only once per parallel rotation stage, not once per rotation. The local basis changes act independently on the submodes, and the phase gates act on disjoint control copies after fanout. Giving each rotation separate synthesis ancillas preserves this parallelism. For fixed $s$, decomposing these primitives introduces only a constant number of rotation stages per shear, and there are six shears. Adding their synthesis depth to the exact input preparation, counting, fanout, and readout depth gives
\begin{align}
D=O\!\left(\log(4mK)+\log\frac{G_{\rm rot}}{\epsilon}\right).
\label{eq:precisiondepth}
\end{align}
Thus, the total rotation count fixes the individual accuracy, whereas the number of sequential rotation stages fixes the depth overhead. At most two fresh ancillas per rotation add $O(G_{\rm rot})$ qubits. Any residual error in these ancillas or in the counting registers is included in the full-state bound; exact ancilla cleanup is not assumed after synthesis.

To collect the bounds, note that for fixed $s$, the chosen $K$ is polynomial in $m$ and $1/\epsilon$; hence, $\log K=O(\log(m/\epsilon))$ and $b_A=O(\log(m/\epsilon))$. Substituting into the rotation count gives $G_{\rm rot}=O(mK+m^2\log^3(m/\epsilon))$ and hence $\log(G_{\rm rot}/\epsilon)=O(\log(m/\epsilon))$. Together with Eq.~\eqref{eq:precisiondepth}, this yields
\begin{align}
D=O(\log(m/\epsilon)),~~ W=O(mK+m^2\log^3(m/\epsilon)).
\label{eq:fullresources}
\end{align}
The $O(mK)$ term includes the submode registers, counting and readout workspace, and synthesis ancillas for the local basis changes. The $O(m^2\log^3(m/\epsilon))$ term includes the control copies and synthesis ancillas for the diagonal phases, as well as the matrix-entry input registers. These bounds give the depth and width claims in Section~\ref{sec:model}.

\begin{proof}[Proof of Theorem~\ref{thm:approx}]
Use the six-shear dilation, the exact submode preparation of the specified Fock input, and the qubit circuits. The three error allocations give norm error of the full state at most $\epsilon$, which also bounds the trace distance. Apply the readout map specified in Section~\ref{sec:readout}: sum submode occupations across each $B0,B1$ output pair, discard the $A$ records, and report $\perp$ for an incorrect total. This map gives $P_U$ on the ideal state. Contractivity under measurement and deterministic processing gives $\|\widetilde P_U-P_U\|_{\TV}\le\epsilon$.

For the theorem, take $s=1$. Then $K=O(m^4\epsilon^{-2})$, and Eq.~\eqref{eq:fullresources} gives width polynomial in $m$ and $1/\epsilon$. Equation~\eqref{eq:fullresources} also gives depth $O(\log(m/\epsilon))$, as claimed. The occupation sums and validity test have logarithmic depth. Appendix~\ref{sec:uniformity} describes a logspace generator for the circuit wiring and gate sequences over the finite gate set at fixed inverse-polynomial precision. Thus, the circuit over a finite gate set, including readout, has the claimed depth, width, and uniformity.
\end{proof}

\begin{proof}[Proof of Corollary~\ref{cor:width}]
Write the fixed target accuracy as $\epsilon=m^{-c}$ with $c>0$. Fix $\alpha>2$ and choose an integer $s>(3+2c)/(\alpha-2)$. This cutoff depends only on the requested accuracy and width exponents, and stays fixed as $m$ grows. With $\epsilon=m^{-c}$ and $n\le m$, the submode choice gives $mK=O(m^{2+(3+2c)/s})$. The other width term in Eq.~\eqref{eq:fullresources} is $O(m^2\log^3m)$; hence, both terms are $O(m^\alpha)$. The depth is $O(\log m)$, and the same uniformity argument for a fixed cutoff applies.
\end{proof}

\section{Implications for decision complexity}\label{sec:decision}
The sampling theorem also has a direct interpretation in decision complexity. If an output event can be recognized by a shallow classical circuit, our sampler can be followed by that circuit to decide whether the event occurs with high or low probability. We formalize this consequence for logspace-uniform $\mathrm{NC}^1$ predicates, obtaining promise decision problems in $\mathrm{BQNC}^1$. The promise below accounts for the finite precision of the interferometer description.

Following the bounded-error convention for shallow quantum circuits~\cite{CleveWatrous}, we use $\mathrm{BQNC}^1$ for promise problems decided with bounded error by a logspace-uniform family of polynomial-width quantum circuits over a fixed finite gate set, with depth $O(\log N)$ on inputs of length $N$. The input is supplied in the computational basis and one designated output bit is measured at the end.

Fix the sampler accuracy at $1/12$ and the associated matrix-entry precision from Section~\ref{sec:resources}. The decoded finite-precision matrix $\widetilde U$ need not itself be exactly unitary. For the resulting input string $x$ of Section~\ref{sec:model}, let $\mathcal U_x$ denote the set of exact unitaries consistent with this precision: $U\in\mathcal U_x$ if every real and imaginary entry of $U$ differs from the corresponding entry of $\widetilde U$ by at most $\eta_A$, with $\eta_A$ chosen in Section~\ref{sec:resources}. The input promise requires $\mathcal U_x\ne\varnothing$. The decision promise below holds uniformly over this set, so the answer is determined by the finite string $x$.

We now state the corresponding decision problem. Let $f(x,\bm n)$ be any logspace-uniform $\mathrm{NC}^1$ predicate~\cite{AroraBarak}, where $x$ denotes the explicit interferometer description and photon number $n$, and $\bm n$ is a reported output pattern. The distribution $P_U$ always uses one photon in each of the first $n$ input modes and vacuum elsewhere. For each compatible unitary $U\in\mathcal U_x$, define
\begin{align}
p_f(x;U)\coloneqq\sum_{\bm n}P_U(\bm n)f(x,\bm n).
\label{eq:decisionacceptance}
\end{align}
Consider the robust promise that either $p_f(x;U)\ge2/3$ for every $U\in\mathcal U_x$ or $p_f(x;U)\le1/3$ for every $U\in\mathcal U_x$.

\begin{corollary}[Decision complexity]\label{cor:decision}
For every predicate $f$ above, deciding which side of this promise holds is in $\mathrm{BQNC}^1$.
\end{corollary}
\begin{proof}
Run a constant number of independent copies of the sampler in parallel. Measuring their output registers would produce independent samples $\bm n^{(1)},\ldots,\bm n^{(r)}$; for each sample, reject on the failure outcome by setting $f(x,\perp)=0$. For one copy, the acceptance probability differs from Eq.~\eqref{eq:decisionacceptance} by at most $1/12$, so it is at least $7/12$ in the yes case and at most $5/12$ in the no case. A majority vote over a constant number of samples therefore gives the usual bounded-error guarantee.

To express this sampling procedure as a $\mathrm{BQNC}^1$ circuit with a single final measurement, defer the measurements and compute the occupation sums, validity checks, and $f$ coherently for each copy. The $\mathrm{NC}^1$ circuit for $f$ can be unfolded into a polynomial-size formula of the same depth, with repeated input bits supplied through parallel CNOT trees~\cite{AroraBarak,MNcodes}. Evaluate these formulas reversibly, compute their majority, and measure only the resulting decision bit. This gives the same acceptance probability as the sampling procedure above, while preserving polynomial width, logarithmic depth, and logspace uniformity.
\end{proof}

\section{Discussion}\label{sec:discussion}
We have shown that the boson-sampling distribution can be simulated to inverse-polynomial total variation error by a logspace-uniform qubit circuit of logarithmic depth and polynomial width. The circuit uses Clifford+$T$ gates and a single final measurement, with the interferometer entries supplied explicitly as input bits. For any fixed inverse-polynomial accuracy, the number of qubits can be made arbitrarily close to quadratic in the number of modes. 

This result addresses the quantum resource question raised in the introduction. The conjectured classical hardness of boson sampling is compatible with a qubit implementation requiring only logarithmically many sequential gate layers, even when the original optical interferometer has much greater depth. Section~\ref{sec:decision} also gives a $\mathrm{BQNC}^1$ upper bound for promise decision problems obtained from output events recognized by uniform $\mathrm{NC}^1$ predicates, provided their probabilities have a constant promise gap. These statements locate boson sampling within shallow quantum computation, but do not establish a separation from universal quantum computation or determine its minimum simulation depth.

Several resource questions remain open. The width and constants may admit substantial improvements, and imposing spatially local qubit interactions may introduce additional depth. It is also natural to ask whether less than logarithmic depth can suffice under the present model. For interferometers specified by an exact finite description, another open question is whether exact boson sampling admits logarithmic depth with polynomially many qubits and efficiently computable parameters for continuous gates.

\begin{acknowledgments}
    Generative AI tools were used to assist with exploratory calculations, literature searches, and language editing. All results, proofs, and references were independently verified by the author, who takes full responsibility for the content of the manuscript.
    This work was supported by the National Research Foundation of Korea Grants (No. RS-2024-00431768 and No. RS-2025-00515456) funded by the Korean government (Ministry of Science and ICT (MSIT)) and the Institute of Information \& Communications Technology Planning \& Evaluation (IITP) Grants funded by the Korean government (MSIT) (No. RS-2024-00437284, No. IITP-2025-RS-2025-02283189 and No. IITP-2025-RS-2025-02263264) by Global Partnership Program of Leading Universities in Quantum Science and Technology (RS-2025-08542968) through the National Research Foundation of Korea~(NRF) funded by the Korean government (Ministry of Science and ICT(MSIT)).
\end{acknowledgments}

\appendix
\section{Three-shear factorization of phase-space rotations}\label{app:shears}
For one canonical pair $[\hat q,\hat p]=i$, write $\hat Q_q(a)=e^{-ia\hat q^2/2}$ and $\hat Q_p(b)=e^{-ib\hat p^2/2}$. The standard decomposition of a planar rotation into three shears~\cite{Paeth1986} gives, at the unitary level,
\begin{align}
e^{-i\theta(\hat q^2+\hat p^2)/2}
=\hat Q_q\!\left(\tan\frac\theta2\right)
\hat Q_p(\sin\theta)
\hat Q_q\!\left(\tan\frac\theta2\right),
\label{eq:rotationunitary}
\end{align}
for $|\theta|<\pi$. Both sides implement the same transformation of the quadratures. For these quadratic unitaries, this determines the operators up to a sign, which is fixed by continuity from $\theta=0$; hence the equality is exact~\cite{MoshinskyQuesne}.

Since $R_U$ is a real symmetric involution, an orthogonal change of basis diagonalizes it with eigenvalues $\pm1$. Applying Eq.~\eqref{eq:rotationunitary} to each eigenmode and transforming back gives
\begin{align}
\hat V_U=\hat Q_q(R_U)\hat Q_p(R_U)\hat Q_q(R_U).
\end{align}
This diagonalization is used only to prove the identity and is not performed in the circuit.

For the encoder, $C=\exp\!\left[-\frac{i\pi}{4}\begin{pmatrix}0&I_m\\I_m&0\end{pmatrix}\right]$ on $A0,A1$. Its eigenmodes have angles $\pm\pi/4$, giving position-shear strengths $\pm(\sqrt2-1)$ and momentum-shear strengths $\pm1/\sqrt2$. Thus,
\begin{align}
A_{\mathrm{enc}}=(2-\sqrt2)B_{\mathrm{enc}},\qquad
B_{\mathrm{enc}}=\frac1{\sqrt2}\begin{pmatrix}0&I_m\\I_m&0\end{pmatrix}\oplus0_{2m},
\end{align}
which reproduces Eq.~\eqref{eq:encodershears}. Since the vacuum phases $e^{-i\theta/2}$ cancel between the opposite angles, the three-shear product implements exactly $C\oplus I_{2m}$.

\section{Proof of the local truncation bound}\label{app:error}
This appendix proves Theorem~\ref{thm:truncationerror}. We first prove the theorem by invoking the estimates established in the subsequent subsections. Throughout this appendix, let $\hat P$ project onto the space in which every submode has occupation at most $s$, and let $\hat P^\perp\coloneqq\hat I-\hat P$.

\subsection{Proof of Theorem~\ref{thm:truncationerror}}
\begin{proof}
We embed the truncated local Fock space in the full $4mK$-mode Fock space by mapping each encoded occupation state to the corresponding Fock state. Since the split input has local occupation at most one, both trajectories start from the same state in $\im\hat P$.

For the error analysis, write each original $4m$-mode shear as $\hat Q_\ell=e^{-i\hat H_\ell}$ and parametrize it by $e^{-i\tau\hat H_\ell}$, where $0\le\tau\le1$. At $t=\ell-1+\tau$, let $\hat G(t)$ be the corresponding propagator on the original $4m$ modes, with $\hat G(0)=\hat I$. Thus, $i\partial_t\hat G(t)=\hat H_\ell\hat G(t)$ within this interval, $\hat G(\ell)$ is the circuit after $\ell$ completed shears, and $\hat G(6)$ is the full six-shear circuit.

The Hamiltonian that actually acts on the $4mK$ submodes during this interval is
\begin{align}
\hat H_K(t)\coloneqq
\hat{\mathcal D}_K
(\hat H_\ell\otimes\hat I_{\rm aux})
\hat{\mathcal D}_K^\dagger,
\qquad t=\ell-1+\tau.
\label{eq:enlargedHamiltonian}
\end{align}
For a $q$-shear, this is the full submode Hamiltonian $\hat H(A)$ used below; the $p$-shear is defined analogously.

Let $\ket\psi$ be the initial state on the $4m$ modes before the submode splitters, including the vacuum blocks introduced in Step 1. The enlarged circuit starts from $\ket{\Psi_{\mathrm{id}}(0)}\coloneqq\hat{\mathcal D}_K(\ket\psi\otimes\ket0_{\rm aux})$. Thus, $t=0$ is after input splitting and before the first encoder shear. Its ideal trajectory is
\begin{align}
\ket{\Psi_{\mathrm{id}}(t)}
&=\hat{\mathcal D}_K(\hat G(t)\otimes \hat I_{\rm aux})\hat{\mathcal D}_K^\dagger\ket{\Psi_{\mathrm{id}}(0)}\nonumber\\
&=\hat{\mathcal D}_K\bigl[\hat G(t)\ket\psi\otimes\ket0_{\rm aux}\bigr].
\label{eq:idealtrajectory}
\end{align}
Differentiating Eq.~\eqref{eq:idealtrajectory} and using Eq.~\eqref{eq:enlargedHamiltonian} gives $i\partial_t\ket{\Psi_{\mathrm{id}}(t)}=\hat H_K(t)\ket{\Psi_{\mathrm{id}}(t)}$. Lemma~\ref{lem:compression}, proved below, shows that local truncation retains exactly the matrix elements of $\hat H_K(t)$ inside the cutoff space. Thus, the embedded finite circuit evolves under $\hat P\hat H_K(t)\hat P$ from the same split input; denote its trajectory by $\ket{\Psi_{\mathrm{tr}}(t)}$. These are precisely the two trajectories compared in Theorem~\ref{thm:truncationerror}.

Both output states are regarded as vectors in the full Fock space. We decompose their difference into components inside and outside the cutoff subspace. Write $\ket{e(t)}\coloneqq\ket{\Psi_{\mathrm{tr}}(t)}-\hat P\ket{\Psi_{\mathrm{id}}(t)}$ for this difference. Here, $\hat P\ket{\Psi_{\mathrm{id}}(t)}$ is the instantaneous, unnormalized projection of the ideal state, whereas $\ket{\Psi_{\mathrm{tr}}(t)}$ evolves under the truncated Hamiltonian. Inserting $\hat P+\hat P^\perp=\hat I$ and applying the triangle inequality gives
\begin{align}
\norm{\ket{\Psi_{\mathrm{tr}}(6)}-\ket{\Psi_{\mathrm{id}}(6)}}
\le \norm{\ket{e(6)}}+\norm{\hat P^\perp\ket{\Psi_{\mathrm{id}}(6)}}.
\label{eq:errorsplit}
\end{align}
The first term is the error within the retained space; the second is the ideal output remaining outside it. We first bound the error within the retained space using the equations of motion, and then control the leakage terms that result.

The two trajectories start from the same state in $\im \hat P$; hence, $\ket{e(0)}=0$. Subtracting their equations of motion gives
\begin{align}
i\partial_t\ket{e(t)}=\hat P\hat H_K(t)\hat P\ket{e(t)}-\hat P\hat H_K(t)\hat P^\perp\ket{\Psi_{\mathrm{id}}(t)}.
\end{align}
The source term $\hat P\hat H_K(t)\hat P^\perp\ket{\Psi_{\mathrm{id}}(t)}$ couples the discarded part of the ideal state back into the retained space. Duhamel's formula expresses this error as the time integral of the source, propagated from each intermediate time to the output. Since $\hat P\hat H_K(t)\hat P$ is Hermitian on the retained space, this propagation is unitary and preserves the norm. Taking norms therefore gives
\begin{align}
\norm{\ket{e(6)}}\le
\int_0^6\norm{\hat P\hat H_K(t)\hat P^\perp\ket{\Psi_{\mathrm{id}}(t)}}\,dt.
\end{align}
Substituting into Eq.~\eqref{eq:errorsplit} yields
\begin{align}
&\norm{\ket{\Psi_{\mathrm{tr}}(6)}-\ket{\Psi_{\mathrm{id}}(6)}}\nonumber\\
&\quad\le\norm{\hat P^\perp\ket{\Psi_{\mathrm{id}}(6)}}+
\int_0^6\norm{\hat P\hat H_K(t)\hat P^\perp\ket{\Psi_{\mathrm{id}}(t)}}\,dt.
\label{eq:duhamel}
\end{align}
The relative Hamiltonian bound in Lemma~\ref{lem:Hrelative} converts the source term into weighted photon-number leakage:
$\norm{\hat P\hat H_K(t)\hat P^\perp\ket{\Psi_{\mathrm{id}}(t)}}\le\sqrt3\norm{(\Nop+4m+2)\hat P^\perp\ket{\Psi_{\mathrm{id}}(t)}}$. The same weighted leakage bounds the endpoint term because $\Nop+4m+2\ge\hat I$. Lemma~\ref{lem:leakage} then shows that uniform splitting suppresses this quantity by $K^{-s/2}$:
\begin{align}
&\norm{(\Nop+4m+2)\hat P^\perp\ket{\Psi_{\mathrm{id}}(t)}}\nonumber\\
&\qquad\le
\frac{\norm{(\Nop+4m+2)^{(s+3)/2}\hat G(t)\ket\psi}}
{\sqrt{(s+1)!}\,K^{s/2}}.
\label{eq:summaryleak}
\end{align}
Here, $\Nop$ acts on the submodes on the left and on the original modes on the right. The remaining numerator is an original-mode photon-number moment. Lemma~\ref{lem:moment} bounds it uniformly throughout all six shears by $[4(n+4m+2s+8)]^{(s+3)/2}$. The endpoint and the integral over six unit intervals in Eq.~\eqref{eq:duhamel} therefore give
\begin{align}
&\norm{\ket{\Psi_{\mathrm{tr}}(6)}-\ket{\Psi_{\mathrm{id}}(6)}}\nonumber\\
&\quad\le\frac{(1+6\sqrt3)2^{s+3}}{\sqrt{(s+1)!}}
\frac{(n+4m+2s+8)^{(s+3)/2}}{K^{s/2}},
\end{align}
as claimed.
\end{proof}

\subsection{Exact local compression}
The following lemma justifies the truncated Hamiltonian used in the theorem proof.

\begin{lemma}[Exact local compression]\label{lem:compression}
For every coefficient matrix produced by Proposition~\ref{prop:six},
\begin{align}
\hat{\bar H}(A)=\hat P\hat H(A)\hat P\big|_{\im \hat P}.
\label{eq:PHP}
\end{align}
The same identity holds for a $p$-shear.
\end{lemma}
\begin{proof}
Let $\hat P_{i\mu}^{(s)}\coloneqq\sum_{r=0}^s\ket r\bra r$ be the local cutoff projector for submode $(i,\mu)$. Because $A_{ii}=0$, every nonzero term in $\hat H(A)$ has $i\ne j$ and therefore acts on two distinct submodes $(i,\mu)$ and $(j,\nu)$. The product cutoff then factors, and, on $\im\hat P$,
\begin{align}
&\hat P\hat q_{i\mu}\hat q_{j\nu}\hat P
\big|_{\im\hat P}\nonumber\\
&\quad=\bigl(\hat P_{i\mu}^{(s)}\hat q_{i\mu}\hat P_{i\mu}^{(s)}\bigr)
\bigl(\hat P_{j\nu}^{(s)}\hat q_{j\nu}\hat P_{j\nu}^{(s)}\bigr)
=\hat{\bar q}_{i\mu}\hat{\bar q}_{j\nu}.
\end{align}
Summing these identities with the coefficients $A_{ij}/(2K)$ gives Eq.~\eqref{eq:PHP}. The same argument with $\hat p$ proves the $p$-shear identity.
\end{proof}

The identity shows that the finite circuit retains exactly the matrix elements of the ideal Hamiltonian within the cutoff space.

\subsection{A Hamiltonian bound on the full submode space}
The source term in Eq.~\eqref{eq:duhamel} applies the shear Hamiltonian to $\hat P^\perp\ket{\Psi_{\mathrm{id}}(t)}$. This projection can excite mode combinations orthogonal to the collective modes. We therefore bound the Hamiltonian on the full submode Fock space.

\begin{lemma}[Relative bound for a shear]\label{lem:Hrelative}
Let $\Nop$ denote total photon number on all $4mK$ submodes. For $\norm A\le1$ and any state $\ket\xi$ in the full $4mK$-mode Fock space for which the right-hand side below is finite,
\begin{align}
\norm{\hat H(A)\ket\xi}\le\sqrt3\norm{(\Nop+4m+2)\ket\xi}.
\label{eq:Hrelative}
\end{align}
The same bound holds for a $p$-shear.
\end{lemma}
\begin{proof}
The shear acts through the normalized collective quadratures
\begin{align}
\hat q_i^{\mathrm{col}}\coloneqq\frac1{\sqrt K}\sum_{\mu=1}^K\hat q_{i\mu},~~
\hat p_i^{\mathrm{col}}\coloneqq\frac1{\sqrt K}\sum_{\mu=1}^K\hat p_{i\mu},
~~ 1\le i\le4m.
\label{eq:collectivequadratures}
\end{align}
Let $\hat\Pi_{\le k}$ project onto total photon number at most $k$. For any state $\ket\varphi=\hat\Pi_{\le k}\ket\varphi$,
\begin{align}
\sum_{i=1}^{4m}\norm{\hat q_i^{\mathrm{col}}\ket\varphi}^2
=\bra\varphi\sum_{i=1}^{4m}(\hat q_i^{\mathrm{col}})^2\ket\varphi
\le(2k+4m)\norm{\ket\varphi}^2.
\label{eq:qcolumncut}
\end{align}
Here, $\hat N_{\rm col}\coloneqq\tfrac12\sum_{i=1}^{4m}[(\hat q_i^{\mathrm{col}})^2+(\hat p_i^{\mathrm{col}})^2-1]$ counts photons in the $4m$ collective modes. The remaining mode combinations contribute a nonnegative photon number, so $\hat N_{\rm col}\le\Nop$. Therefore, $\sum_i(\hat q_i^{\mathrm{col}})^2\le2\hat N_{\rm col}+4m\le2\Nop+4m$, which gives Eq.~\eqref{eq:qcolumncut} on $\im\hat\Pi_{\le k}$.

To bound the quadratic Hamiltonian, take normalized $\ket\varphi$ and $\ket\chi$ with total number at most $k$ and $k+2$, respectively. Cauchy--Schwarz over the mode indices gives
\begin{align}
&\abs{\bra\chi\hat H(A)\ket\varphi}
=\frac12\abs{\sum_{i,j=1}^{4m}A_{ij}\bra\chi\hat q_i^{\mathrm{col}}\hat q_j^{\mathrm{col}}\ket\varphi}\\
&\le\frac{\norm A}{2}
\left(\sum_{i=1}^{4m}\norm{\hat q_i^{\mathrm{col}}\ket\chi}^2\right)^{1/2}
\left(\sum_{j=1}^{4m}\norm{\hat q_j^{\mathrm{col}}\ket\varphi}^2\right)^{1/2}.
\end{align}
The Hamiltonian changes total photon number by at most two; hence, $\hat H(A)\ket\varphi$ lies in $\im\hat\Pi_{\le k+2}$. Taking the supremum over these two normalized states and applying Eq.~\eqref{eq:qcolumncut} at $k$ and $k+2$ gives
\begin{align}
\norm{\hat H(A)\hat\Pi_{\le k}}
&\le\frac12\sqrt{(2k+4m+4)(2k+4m)}\nonumber\\
&\le k+2m+1\le k+4m+2.
\label{eq:Hcut}
\end{align}
Now decompose $\ket\xi=\sum_{k=0}^\infty\ket{\xi_k}$ into sectors of fixed total photon number. A quadratic shear changes number only by $0,+2,-2$. At most three input sectors contribute to any given output sector. Cauchy--Schwarz in that sum gives
\begin{align}
\norm{\hat H(A)\ket\xi}^2
&\le3\sum_{k=0}^\infty\norm{\hat H(A)\ket{\xi_k}}^2 \\
&\le3\sum_{k=0}^\infty(k+4m+2)^2\norm{\ket{\xi_k}}^2.
\end{align}
This is Eq.~\eqref{eq:Hrelative}.
\end{proof}

\subsection{Dilution suppresses occupations above the cutoff}
The splitting isometry $\hat V_K$ attaches $K-1$ vacuum modes to each original mode and applies the splitting network:
$\hat V_K\ket\varphi\coloneqq\hat{\mathcal D}_K(\ket\varphi\otimes\ket0_{\rm aux})$.
Within each block, only its collective mode can be occupied; the orthogonal mode combinations remain in vacuum. On these split states, each physical submode annihilation operator contributes a factor $K^{-1/2}$, as shown below. The $r$th factorial moment, which involves $r$ annihilations and $r$ creations in the same submode, is therefore reduced by $K^{-r}$.

Let $\hat n_{i\mu}\coloneqq\hat a_{i\mu}^\dagger\hat a_{i\mu}$ be the number operator of submode $(i,\mu)$ and $\hat n_i\coloneqq\hat a_i^\dagger\hat a_i$ that of the original mode. We write
$\fall{x}{r}\coloneqq x(x-1)\cdots(x-r+1)$. For an occupation number $x$, this falling factorial counts ordered choices of $r$ photons and vanishes when $x<r$.

\begin{lemma}[Dilution identity]\label{lem:factorial}
For every integer $r\ge1$,
\begin{align}
\hat V_K^\dagger\fall{\hat n_{i\mu}}r\hat V_K
=K^{-r}\fall{\hat n_i}r.
\label{eq:factorial}
\end{align}
More generally, if $f(\Nop)$ is a function of total photon number for which the expressions are defined, then
\begin{align}
\hat V_K^\dagger f(\Nop)\fall{\hat n_{i\mu}}r\hat V_K
=K^{-r}f(\Nop)\fall{\hat n_i}r.
\label{eq:weightedfactorial}
\end{align}
\end{lemma}
\begin{proof}
For each original mode $i$, define the collective output annihilation operator
\begin{align}
\hat a_i^{\mathrm{col}}
\coloneqq\frac1{\sqrt K}\sum_{\mu=1}^K\hat a_{i\mu}.
\label{eq:collectiveannihilator}
\end{align}
Complete it to an orthonormal mode basis with annihilation operators
$\hat a_{i\alpha}^{\perp}$, $\alpha=2,\ldots,K$. In this basis, each physical submode operator decomposes as
\begin{align}
\hat a_{i\mu}
=\frac1{\sqrt K}\hat a_i^{\mathrm{col}}
+\sum_{\alpha=2}^K c_{\mu\alpha}\hat a_{i\alpha}^{\perp}.
\label{eq:submodedecomposition}
\end{align}
The vacuum condition holds for the mode combinations orthogonal to the collective mode, rather than for the individual physical submodes. Indeed, in the collective--orthogonal basis, the image of $\hat V_K$ has the form
\begin{align}
\hat V_K\ket\varphi
=\ket\varphi_{\mathrm{col}}\otimes\ket0_{\perp}.
\label{eq:imageVK}
\end{align}
It follows that
\begin{align}
\hat a_i^{\mathrm{col}}\hat V_K
=\hat V_K\hat a_i,
\qquad
\hat a_{i\alpha}^{\perp}\hat V_K=0.
\label{eq:collectiveintertwining}
\end{align}
Substituting Eq.~\eqref{eq:collectiveintertwining} into Eq.~\eqref{eq:submodedecomposition} removes every orthogonal-mode term and gives
\begin{align}
\hat a_{i\mu}\hat V_K
=\frac1{\sqrt K}\hat V_K\hat a_i.
\label{eq:submodeintertwining}
\end{align}
Applying this relation $r$ times yields
\begin{align}
\hat a_{i\mu}^r\hat V_K
=K^{-r/2}\hat V_K\hat a_i^r.
\label{eq:repeatedintertwining}
\end{align}
Hence,
\begin{align}
\hat V_K^\dagger
(\hat a_{i\mu}^\dagger)^r\hat a_{i\mu}^r
\hat V_K
&=K^{-r}(\hat a_i^\dagger)^r
\hat V_K^\dagger\hat V_K\hat a_i^r\nonumber\\
&=K^{-r}(\hat a_i^\dagger)^r\hat a_i^r,
\end{align}
which is Eq.~\eqref{eq:factorial}, because
$(\hat a^\dagger)^r\hat a^r=\fall{\hat n}r$.
Finally, $\hat V_K$ preserves total photon number, and total photon number commutes with every occupation operator. This gives Eq.~\eqref{eq:weightedfactorial}.
\end{proof}

We now convert the factorial-moment identity into a cutoff bound. If a submode lies outside the retained local space, its occupation is at least $s+1$. For every nonnegative integer $x$,
\begin{align}
{\bf1}_{x\ge s+1}
\le\frac{\fall{x}{s+1}}{(s+1)!},
\label{eq:scalarcutoff}
\end{align}
because both sides vanish for $x\le s$, while
$\fall{x}{s+1}\ge(s+1)!$ for $x\ge s+1$.
Taking the union over all $4mK$ submodes therefore gives the diagonal operator inequality
\begin{align}
\hat P^\perp
\le\sum_{i=1}^{4m}\sum_{\mu=1}^K
{\bf1}_{\hat n_{i\mu}\ge s+1}
\le\frac1{(s+1)!}
\sum_{i=1}^{4m}\sum_{\mu=1}^K
\fall{\hat n_{i\mu}}{s+1}.
\label{eq:union}
\end{align}
Lemma~\ref{lem:factorial} contributes $K^{-(s+1)}$ for each term, whereas the sum over $\mu$ contributes a factor $K$. This is the origin of the net factor $K^{-s}$ below.

Let $\Nop$ denote total photon number, either on the $4mK$ submodes or, after conjugation by $\hat V_K$, on the original $4m$ modes.

\begin{lemma}[Weighted leakage]\label{lem:leakage}
For $c\ge1$ and any state $\ket\varphi$ on the original $4m$ modes in the domain of the displayed moment,
\begin{align}
\norm{(\Nop+c)\hat P^\perp\hat V_K\ket\varphi}
\le\frac{\norm{(\Nop+c)^{(s+3)/2}\ket\varphi}}
{\sqrt{(s+1)!}\,K^{s/2}}.
\label{eq:leakage}
\end{align}
\end{lemma}
\begin{proof}
Since $\hat P^\perp$ is a projector and commutes with $\Nop$, squaring the left-hand side gives
\begin{align}
\norm{(\Nop+c)\hat P^\perp\hat V_K\ket\varphi}^2
=\bra\varphi\hat V_K^\dagger
(\Nop+c)^2\hat P^\perp\hat V_K\ket\varphi.
\end{align}
Apply Eq.~\eqref{eq:union}, followed by the weighted identity
Eq.~\eqref{eq:weightedfactorial} with $r=s+1$, to obtain
\begin{align}
\norm{(\Nop+c)\hat P^\perp\hat V_K\ket\varphi}^2
\le\frac{\bra\varphi(\Nop+c)^2
\sum_{i=1}^{4m}\fall{\hat n_i}{s+1}\ket\varphi}
{(s+1)!K^s}.
\label{eq:weightedleakageintermediate}
\end{align}
On every occupation basis vector,
\begin{align}
\sum_{i=1}^{4m}\fall{\hat n_i}{s+1}
\le\Nop^{s+1}
\le(\Nop+c)^{s+1}.
\end{align}
Substituting this bound into Eq.~\eqref{eq:weightedleakageintermediate} and taking a square root gives Eq.~\eqref{eq:leakage}.
\end{proof}

\subsection{Photon number moments during the six shears}
We first bound the matrices describing the action of the shears on the quadratures, then use this estimate to control the photon number moments in Lemma~\ref{lem:leakage}.
For a unitary $\hat{\mathcal U}$ in the ideal evolution, describe its action on the $8m$ quadratures by the real matrix defined through $S(\hat{\mathcal U})(\hat q,\hat p)^{\T}\coloneqq \hat{\mathcal U}^\dagger(\hat q,\hat p)^{\T}\hat{\mathcal U}$. The shear conjugation rules from Step 1 give
\begin{align}
S_q(A)=\begin{pmatrix}I&0\\-A&I\end{pmatrix},\qquad
S_p(B)=\begin{pmatrix}I&B\\0&I\end{pmatrix}.
\label{eq:shearmatrices}
\end{align}
\begin{lemma}[Uniform bound throughout the ideal evolution]\label{lem:amplification}
At every intermediate time in the ideal six-shear evolution, the matrix describing its action on the quadratures has operator norm at most two.
\end{lemma}
\begin{proof}
Consider either of the two three-shear factorizations, and let
\begin{align}
S_{\rm fin}\coloneqq S_q(A)S_p(B)S_q(A)
\end{align}
denote its completed transformation, which is orthogonal. Since $\norm A\le1$, the triangle inequality gives $\norm{S_q(tA)}\le1+t\norm A\le2$ for $0\le t\le1$.

While the first, second, and third shears are being applied, respectively, the intermediate transformations are
\begin{align}
&S_q(tA),\\
&S_p(tB)S_q(A)
=(1-t)S_q(A)+tS_q(-A)S_{\rm fin},\\
&S_q(tA)S_p(B)S_q(A)
=S_q((t-1)A)S_{\rm fin},
\qquad 0\le t\le1.
\end{align}
The first expression has norm at most two directly. The third has norm at most two because $S_{\rm fin}$ is orthogonal and $|t-1|\le1$. For the second, the displayed identity writes it as a convex combination of $S_q(A)$ and $S_q(-A)S_{\rm fin}$, both of which have norm at most two. Hence, all three intermediate transformations have norm at most two. Any transformation completed before the current three-shear factorization is passive and therefore orthogonal, so multiplying by it does not change these norms. The bound thus holds throughout all six shears.
\end{proof}

To obtain the photon number moment bound, fix any intermediate time and let $\hat G\coloneqq\hat G(t)$ be the corresponding ideal evolution on the original $4m$ modes, with $S$ the matrix defined above. Put
\begin{align}
\hat C\coloneqq\Nop+4m+2,\qquad \hat C_G\coloneqq\hat G^\dagger \hat C\hat G.
\end{align}
Here, $\hat C_G$ is a positive operator. With $\hat x\coloneqq(\hat q,\hat p)^{\T}$,
\begin{align}
\hat C_G=\frac12\hat x^{\T}S^{\T}S\hat x+2m+2.
\label{eq:conjugatedN}
\end{align}

\begin{lemma}[Uniform moment estimate]\label{lem:moment}
Let $r\coloneqq s+3$. For every intermediate ideal evolution $\hat G=\hat G(t)$,
\begin{align}
\norm{\hat C^{r/2}\hat G \hat\Pi_{\le n}}
\le\left[4(n+4m+2s+8)\right]^{r/2}.
\label{eq:moment}
\end{align}
\end{lemma}
\begin{proof}
For $\ket\varphi=\hat\Pi_{\le k}\ket\varphi$, the full set of quadratures satisfies $\sum_{j=1}^{8m}\norm{\hat x_j\ket\varphi}^2\le(2k+4m)\norm{\ket\varphi}^2$. Since $\norm{S^{\T}S}\le4$, the same Cauchy--Schwarz estimate used for Eq.~\eqref{eq:Hcut}, now applied to Eq.~\eqref{eq:conjugatedN}, gives
\begin{align}
\norm{\hat C_G\hat\Pi_{\le k}}
&\le2\sqrt{(2k+4m+4)(2k+4m)}+2m+2
\nonumber\\&\le4(k+4m+2).
\label{eq:CGbound}
\end{align}

The quadratic operator $\hat C_G$ changes photon number by at most two. For any integer $r\ge1$, iterating Eq.~\eqref{eq:CGbound} therefore gives
\begin{align}
\norm{\hat C_G^r\hat\Pi_{\le n}}
\le4^r\prod_{j=0}^{r-1}(n+4m+2j+2).
\label{eq:momentiteration}
\end{align}
Moreover,
\begin{align}
\norm{\hat C^{r/2}\hat G\hat\Pi_{\le n}}^2
=\norm{\hat\Pi_{\le n}\hat G^\dagger \hat C^r\hat G\hat\Pi_{\le n}}
=\norm{\hat\Pi_{\le n}\hat C_G^r\hat\Pi_{\le n}}.
\end{align}
For $r=s+3$, every factor in Eq.~\eqref{eq:momentiteration} is at most $n+4m+2s+8$. This proves Eq.~\eqref{eq:moment}. The identity above also covers odd $r$.
\end{proof}
Finite Fock vectors acted on by a bounded Gaussian transformation on finitely many modes have the required moments. The estimates can first be established on cores with finite photon number and then extended by norm limits, justifying the derivatives and integrals in the trajectory comparison above.

\section{Circuit implementation and uniformity}\label{sec:uniformity}
This appendix gives the reversible arithmetic, parallelization of diagonal phases, and logspace generation used in Proposition~\ref{prop:shearresource}. The cutoff $s$ is fixed throughout, and constants hidden in $O(\cdot)$ may depend on it.

\subsection{Reversible histogram counting}
For each block $i$ and eigenvalue label $r$, compare all $K$ local labels with $r$ and coherently add the indicator bits. A quantum carry-save tree reduces these bits to two $O(\log K)$-bit numbers in $O(\log K)$ constant-depth rounds~\cite{Gossett}. A final reversible ripple-carry addition preserves the $O(\log K)$ depth~\cite{Cuccaro}, with all blocks and labels processed in parallel. Copying the resulting counts into fresh registers and reversing the intermediate arithmetic~\cite{Bennett1973} gives the clean map
\begin{align}
\ket{\{r_{i\mu}\}}_D\ket0_H\longmapsto
\ket{\{r_{i\mu}\}}_D\ket{\{h_{i,r}\}}_H.
\end{align}
The number of partial words decreases geometrically, while their length grows by at most one bit per round. At round $\ell$, there are $O(K(2/3)^\ell+1)$ words of length $O(\ell+1)$. Their total storage over $O(\log K)$ rounds is $O(K\sum_{\ell\ge0}(\ell+1)(2/3)^\ell+\log^2 K)=O(K)$, including the intermediate results retained to reverse the computation. Since the number of labels is fixed, all histograms require $O(4mK)$ auxiliary qubits.

\subsection{Parallel diagonal phase}
To implement the phase in Eq.~\eqref{eq:histphase}, expand the eigenvalue counts and the coefficients $A_{ij}$ in binary. We do not compute the full value $E_A(\bm h)$ into an energy register. Instead, its binary expansion factors $e^{-iE_A(\bm h)}$ into controlled phase gates whose fixed angles contain the local eigenvalues and binary place values. In two's-complement notation each coefficient is a linear combination of its bits with fixed signed weights; the sign bit therefore enters in the same way as the other bits.

More explicitly, put $L=\lceil\log_2(K+1)\rceil$ and write $A_{ij}=\sum_{c=0}^{b_A-1}w_c a_{ij,c}$ and $h_{i,r}=\sum_{\alpha=0}^{L-1}2^\alpha h_{i,r,\alpha}$, where $w_c$ includes the negative two's-complement sign weight. Then
\begin{align}
E_A(\bm h)
=\sum_{i,j,r,t,c,\alpha,\beta}
\frac{w_c2^{\alpha+\beta}\lambda_r\lambda_t}{2K}
a_{ij,c}h_{i,r,\alpha}h_{j,t,\beta}.
\label{eq:binaryphase}
\end{align}
Each summand gives a diagonal gate on at most three bits, with an angle fixed by the circuit indices and the bits acting as quantum controls. Thus, there are at most $M\coloneqq(4m)^2(s+1)^2b_A L^2$ phase terms. Repeated bit labels only lower the arity, and each gate has a constant-size decomposition into one-qubit gates and CNOTs~\cite{Barenco}. For hardwired coefficients, their bit expansion can be absorbed into the fixed angles, leaving $O((4m)^2\log^2K)$ terms.

Applying the diagonal-gate parallelization of Ref.~\cite[Proposition 5]{MNcodes}, CNOT trees supply disjoint controls for all phase terms. At most three fresh control wires per term suffice, so the total fanout width is $O(M)$ and its depth is $O(\log(M+1))$. The phase gates then act in parallel. Since they leave the control bits unchanged, reversing the fanout and histogram circuits returns all temporary registers to zero while preserving the phase. Only bounded-arity CNOT gates are used.

\begin{proof}[Proof of Proposition~\ref{prop:shearresource}]
The local basis changes cost $O(1)$ depth. Eigenvalue counting costs $O(\log K)$ depth and $O(4mK)$ auxiliary qubits. Fanout to the polynomially many phase terms and its reversal cost $O(\log(4mK)+\log(b_A+1))$ depth. The phase gates are disjoint after copying. The number of auxiliary qubits includes the counting registers and a constant number of auxiliary qubits for each expanded term. Every auxiliary computation is reversed after the diagonal phase, which proves Eq.~\eqref{eq:shearresource}.
\end{proof}

\subsection{Count of rotations requiring synthesis}
We now count the one-qubit rotations to which the finite-gate approximation in Section~\ref{sec:resources} is applied. Let $g_{\mathcal W}(s)$ be the number of such rotations in a decomposition of one local basis change. This number is constant because $s$ is fixed. Each shear applies a basis change and its inverse to all $4mK$ submodes. Across the six shears, these operations therefore contribute at most
\begin{align}
6\times2\times(4mK)g_{\mathcal W}(s)=O(4mK)
\end{align}
rotation occurrences. The factor of two includes the return to the occupation basis.

For the diagonal phase, Eq.~\eqref{eq:binaryphase} contains at most $(4m)^2(s+1)^2b_A L^2$ constant-arity phase terms per shear, where $L=\lceil\log_2(K+1)\rceil$. Each such term has a constant-size decomposition containing $O(1)$ one-qubit rotations. Multiplication by the six shears changes only the constant factor. Hence, the total number of rotations requiring synthesis is
\begin{align}
G_{\rm rot}=O\!\left(4mK+(4m)^2b_A\log^2K\right).
\label{eq:rotationcountappendix}
\end{align}
The forward and inverse histogram circuits, the fanout trees and their reversals, reversible arithmetic, and input splitting use exact Clifford+$T$ primitives. They affect the gate and qubit counts but add no rotation occurrences to $G_{\rm rot}$.

\subsection{Logspace generation}
We prove the uniformity claim in Section~\ref{sec:model} by generating the gate list from the size parameters and a fixed inverse-polynomial accuracy. The cutoff $s$ is fixed for the family. The generator uses $O(\log m)$ workspace in each of the following tasks.

\begin{enumerate}
\item \emph{Compute the fixed local gates.} This task produces the input-independent local basis changes and rotation angles used in each truncated shear. For the diagonal phase, expanding Eq.~\eqref{eq:histphase} in binary gives constant-arity gates controlled on coefficient and count bits. Their angles depend on bit positions and the eigenvalues of $\hat{\bar q}$. Let $p=\Theta(\log(G_{\rm rot}/\epsilon))$ be the working precision required by the per-rotation error allocation in Section~\ref{sec:resources}. For fixed inverse-polynomial accuracy, Eq.~\eqref{eq:rotationcountappendix} gives $p=O(\log m)$. The working precision includes $O(\log K)$ extra bits to account for the $O(K)$ factors in Eq.~\eqref{eq:binaryphase}, without changing this scaling. Since $s$ is fixed, the eigenvalues of $\hat{\bar q}$ and a choice of eigenvectors are fixed algebraic numbers. Their $p$-bit approximations, and the elementary rotation matrices obtained from them, can be computed in $\poly(p)$ time and $O(p)$ workspace using root isolation for polynomials of fixed degree and arithmetic. Standard fixed-precision algorithms compute the required trigonometric values within the same time and workspace bounds~\cite{BrentZimmermann}. Repeating these computations for the polynomially many gate occurrences therefore takes $\poly(m)$ time and $O(\log m)$ workspace.
\item \emph{Generate the Clifford+$T$ words.} This task replaces each one-qubit rotation arising from the local basis changes and constant-arity phase gates by a word over the fixed gate set. At the same precision $p$, the synthesis theorem of Ref.~\cite{KMM} guarantees a word of length $O(p)$ on a constant number of qubits. Enumerate such words and store each matrix of constant dimension exactly. Each entry is an integer linear combination of $1,\sqrt2,i,i\sqrt2$ divided by a power of two. Multiplying by one elementary gate increases the denominator exponent and the coefficient bit lengths by at most a constant. A word of length $O(p)$ thus needs only $O(p)$ bits for its matrix. For a candidate word $W$ and target elementary rotation $R$, compare the restricted maps $W(\,\cdot\,\otimes\ket0_A)$ and $R(\,\cdot\,)\otimes\ket0_A$ in operator norm. The data--ancilla dimension is constant, so this is a constant-dimensional matrix test using $O(p)$ workspace. A threshold stricter than the required error provides a safety margin for the finite-precision comparison and ensures that a suitable word is accepted. There are $2^{O(p)}=\poly(m)$ candidates, and the word, counters, and matrix entries occupy $O(p)$ bits. Since the norm includes all output rows, it also counts any error in the synthesis ancillas, as required by the error bound in Section~\ref{sec:resources}.
\item \emph{Generate the wiring.} This task outputs the wire indices for the histogram, fanout, input-preparation, and readout circuits. The histogram-counting circuit above and the fanout circuit for the diagonal phase have regular tree structures. At each carry-save round of the histogram circuit, the number and padded bit length of the partial sums are determined solely by the round number and $K$. Gate and wire indices can be enumerated with counters of logarithmic size; forwarded wires can be traced backward by recomputation. For the common layout with $n$ supplied in binary, choose $K$ with $n=m$ in Eq.~\eqref{eq:Kchoice}; since $n\le m$, this has the same asymptotic scaling as the choice in the main text. Use a fixed integer upper bound on $C_s$ and a dyadic accuracy within a constant factor of the target. Testing successive powers of two after raising the inequality to the power $s$ uses only integer arithmetic on $O(\log m)$ bits. The comparisons $j\le n$ for input preparation, together with their fanout and reversal, have regular wiring generated from $m$ and the bit positions of $n$. The readout sums and validity test use the same arithmetic constructions.
\end{enumerate}
These procedures print the polynomial-size circuit using logarithmic workspace. The numerical values of the interferometer entries are supplied to its input wires.

\bibliographystyle{apsrev4-2}
\bibliography{reference}

@article{AA,
 author={Aaronson, Scott and Arkhipov, Alex}, title={The Computational Complexity of Linear Optics}, journal={Theory of Computing}, volume={9}, pages={143--252}, year={2013}, doi={10.4086/toc.2013.v009a004}, eprint={1011.3245}, archivePrefix={arXiv}}

@misc{MN,
 author={Moore, Cristopher and Nilsson, Martin}, title={Some Notes on Parallel Quantum Computation}, year={1998}, eprint={quant-ph/9804034}, archivePrefix={arXiv}}

@misc{Gossett,
 author={Gossett, Phil}, title={Quantum Carry-Save Arithmetic}, year={1998}, eprint={quant-ph/9808061}, archivePrefix={arXiv}}

@book{AroraBarak, author={Arora, Sanjeev and Barak, Boaz}, title={Computational Complexity: A Modern Approach}, publisher={Cambridge University Press}, year={2009}, doi={10.1017/CBO9780511804090}}

@misc{MNcodes,
 author={Moore, Cristopher and Nilsson, Martin}, title={Parallel Quantum Computation and Quantum Codes}, year={1998}, eprint={quant-ph/9808027}, archivePrefix={arXiv}}

@article{KMM,
 author={Kliuchnikov, Vadym and Maslov, Dmitri and Mosca, Michele}, title={Asymptotically Optimal Approximation of Single-Qubit Unitaries by {Clifford} and {T} Circuits Using a Constant Number of Ancillary Qubits}, journal={Physical Review Letters}, volume={110}, pages={190502}, year={2013}, doi={10.1103/PhysRevLett.110.190502}, eprint={1212.0822}, archivePrefix={arXiv}}

@article{Tong,
 author={Tong, Yu and Albert, Victor V. and McClean, Jarrod R. and Preskill, John and Su, Yuan}, title={Provably accurate simulation of gauge theories and bosonic systems}, journal={Quantum}, volume={6}, pages={816}, year={2022}, eprint={2110.06942}, archivePrefix={arXiv}}

@article{OhReview, author={Oh, Changhun}, title={Recent Theoretical and Experimental Progress on Boson Sampling}, journal={Current Optics and Photonics}, volume={9}, number={1}, pages={1--18}, year={2025}, doi={10.3807/COPP.2025.9.1.1}}

@article{OhGraph, author={Oh, Changhun and Lim, Youngrong and Fefferman, Bill and Jiang, Liang}, title={Classical simulation of boson sampling based on graph structure}, journal={Physical Review Letters}, volume={128}, pages={190501}, year={2022}, doi={10.1103/PhysRevLett.128.190501}, eprint={2110.01564}, archivePrefix={arXiv}}

@misc{OhNoisy, author={Oh, Changhun and Jiang, Liang and Fefferman, Bill}, title={On classical simulation algorithms for noisy {Boson Sampling}}, year={2023}, eprint={2301.11532}, archivePrefix={arXiv}}

@article{OhGBS, author={Oh, Changhun and Liu, Minzhao and Alexeev, Yuri and Fefferman, Bill and Jiang, Liang}, title={Classical algorithm for simulating experimental {Gaussian} boson sampling}, journal={Nature Physics}, volume={20}, pages={1461--1468}, year={2024}, doi={10.1038/s41567-024-02535-8}, eprint={2306.03709}, archivePrefix={arXiv}}

@article{KLM, author={Knill, Emanuel and Laflamme, Raymond and Milburn, G. J.}, title={A scheme for efficient quantum computation with linear optics}, journal={Nature}, volume={409}, pages={46--52}, year={2001}, doi={10.1038/35051009}, eprint={quant-ph/0006088}, archivePrefix={arXiv}}

@article{Hamilton, author={Hamilton, Craig S. and Kruse, Regina and Sansoni, Linda and Barkhofen, Sonja and Silberhorn, Christine and Jex, Igor}, title={Gaussian Boson Sampling}, journal={Physical Review Letters}, volume={119}, pages={170501}, year={2017}, doi={10.1103/PhysRevLett.119.170501}, eprint={1612.01199}, archivePrefix={arXiv}}

@article{Zhong, author={Zhong, Han-Sen and Wang, Hui and Deng, Yu-Hao and Chen, Ming-Cheng and Peng, Li-Chao and Luo, Yi-Han and Qin, Jian and Wu, Dian and Ding, Xing and Hu, Yi and Hu, Peng and Yang, Xiao-Yan and Zhang, Wei-Jun and Li, Hao and Li, Yuxuan and Jiang, Xiao and Gan, Lin and Yang, Guangwen and You, Lixing and Wang, Zhen and Li, Li and Liu, Nai-Le and Lu, Chao-Yang and Pan, Jian-Wei}, title={Quantum computational advantage using photons}, journal={Science}, volume={370}, pages={1460--1463}, doi={10.1126/science.abe8770}, year={2020}, eprint={2012.01625}, archivePrefix={arXiv}}

@misc{CC, author={Clifford, Peter and Clifford, Rapha\"{e}l}, title={The Classical Complexity of Boson Sampling}, year={2017}, eprint={1706.01260}, archivePrefix={arXiv}}

@article{CCfast, author={Clifford, Peter and Clifford, Rapha\"{e}l}, title={Faster classical boson sampling}, journal={Physica Scripta}, volume={99}, pages={065121}, year={2024}, doi={10.1088/1402-4896/ad4688}, eprint={2005.04214}, archivePrefix={arXiv}}

@article{Neville, author={Neville, Alex and Sparrow, Chris and Clifford, Rapha\"{e}l and Johnston, Eric and Birchall, Patrick M. and Montanaro, Ashley and Laing, Anthony}, title={Classical boson sampling algorithms with superior performance to near-term experiments}, journal={Nature Physics}, volume={13}, pages={1153--1157}, year={2017}, doi={10.1038/nphys4270}, eprint={1705.00686}, archivePrefix={arXiv}}

@article{HS, author={H{\o}yer, Peter and \v{S}palek, Robert}, title={Quantum fan-out is powerful}, journal={Theory of Computing}, volume={1}, pages={81--103}, year={2005}, doi={10.4086/toc.2005.v001a005}, eprint={quant-ph/0208043}, archivePrefix={arXiv}}

@article{Barenco, author={Barenco, Adriano and Bennett, Charles H. and Cleve, Richard and DiVincenzo, David P. and Margolus, Norman and Shor, Peter and Sleator, Tycho and Smolin, John A. and Weinfurter, Harald}, title={Elementary gates for quantum computation}, journal={Physical Review A}, volume={52}, pages={3457}, year={1995}, doi={10.1103/PhysRevA.52.3457}, eprint={quant-ph/9503016}, archivePrefix={arXiv}}

@article{Ross, author={Ross, Neil J. and Selinger, Peter}, title={Optimal ancilla-free {Clifford+T} approximation of z-rotations}, journal={Quantum Information and Computation}, volume={16}, pages={901--953}, year={2016}, eprint={1403.2975}, archivePrefix={arXiv}}

@article{Sawaya, author={Sawaya, Nicolas P. D. and Menke, Tim and Kyaw, Thi Ha and Johri, Sonika and Aspuru-Guzik, Al\'{a}n and Guerreschi, Gian Giacomo}, title={Resource-efficient digital quantum simulation of {$d$}-level systems for photonic, vibrational, and spin-{$s$} {Hamiltonians}}, journal={npj Quantum Information}, volume={6}, pages={49}, year={2020}, doi={10.1038/s41534-020-0278-0}, eprint={1909.12847}, archivePrefix={arXiv}}

@misc{DawsonNielsen, author={Dawson, Christopher M. and Nielsen, Michael A.}, title={The {Solovay--Kitaev} algorithm}, year={2005}, eprint={quant-ph/0505030}, archivePrefix={arXiv}}

@article{GrierGBS, author={Grier, Daniel and Brod, Daniel J. and Arrazola, Juan Miguel and Alonso, Marcos Benicio de Andrade and Quesada, Nicol\'{a}s}, title={The Complexity of Bipartite {Gaussian} Boson Sampling}, journal={Quantum}, volume={6}, pages={863}, year={2022}, doi={10.22331/q-2022-11-28-863}}

@article{HangleiterEisert, author={Hangleiter, Dominik and Eisert, Jens}, title={Computational advantage of quantum random sampling}, journal={Reviews of Modern Physics}, volume={95}, pages={035001}, year={2023}, doi={10.1103/RevModPhys.95.035001}}

@article{Zhong2021, author={Zhong, Han-Sen and Deng, Yu-Hao and Qin, Jian and Wang, Hui and Chen, Ming-Cheng and others}, title={Phase-Programmable {Gaussian} Boson Sampling Using Stimulated Squeezed Light}, journal={Physical Review Letters}, volume={127}, pages={180502}, year={2021}, doi={10.1103/PhysRevLett.127.180502}}

@article{Deng2023, author={Deng, Yu-Hao and Gu, Yi-Chao and Liu, Hua-Liang and Gong, Si-Qiu and Su, Hao and others}, title={{Gaussian} Boson Sampling with Pseudo-Photon-Number-Resolving Detectors and Quantum Computational Advantage}, journal={Physical Review Letters}, volume={131}, pages={150601}, year={2023}, doi={10.1103/PhysRevLett.131.150601}}

@article{Liu2026, author={Liu, Hua-Liang and Su, Hao and Deng, Yu-Hao and Gong, Si-Qiu and Gu, Yi-Chao and others}, title={{Gaussian} boson sampling with 1,024 squeezed states in 8,176 modes}, journal={Nature}, volume={653}, pages={687--692}, year={2026}, doi={10.1038/s41586-026-10523-6}}

@article{Madsen2022, author={Madsen, Lars S. and Laudenbach, Fabian and Askarani, Mohsen Falamarzi and Rortais, Fabien and Vincent, Trevor and Bulmer, Jacob F. F. and Miatto, Filippo M. and Neuhaus, Leonhard and Helt, Lukas G. and Collins, Matthew J. and Lita, Adriana E. and Gerrits, Thomas and Nam, Sae Woo and Vaidya, Varun D. and Menotti, Matteo and Dhand, Ish and Vernon, Zachary and Quesada, Nicol\'{a}s and Lavoie, Jonathan}, title={Quantum computational advantage with a programmable photonic processor}, journal={Nature}, volume={606}, pages={75--81}, year={2022}, doi={10.1038/s41586-022-04725-x}}

@article{Young2024, author={Young, Aaron W. and Geller, Shawn and Eckner, William J. and Schine, Nathan and Glancy, Scott and Knill, Emanuel and Kaufman, Adam M.}, title={An atomic boson sampler}, journal={Nature}, volume={629}, pages={311--316}, year={2024}, doi={10.1038/s41586-024-07304-4}}

@article{Quesada2022, author={Quesada, Nicol\'{a}s and Chadwick, Rachel S. and Bell, Bryn A. and Arrazola, Juan Miguel and Vincent, Trevor and Qi, Haoyu and Garc\'{i}a-Patr\'{o}n, Ra\'{u}l}, title={Quadratic Speed-Up for Simulating {Gaussian} Boson Sampling}, journal={PRX Quantum}, volume={3}, pages={010306}, year={2022}, doi={10.1103/PRXQuantum.3.010306}}

@article{OhConstant, author={Oh, Changhun}, title={Classical simulability of constant-depth linear-optical circuits with noise}, journal={npj Quantum Information}, volume={11}, pages={126}, year={2025}, doi={10.1038/s41534-025-01041-w}}

@article{Moylett, author={Moylett, Alexandra E. and Turner, Peter S.}, title={Quantum simulation of partially distinguishable boson sampling}, journal={Physical Review A}, volume={97}, pages={062329}, year={2018}, doi={10.1103/PhysRevA.97.062329}}

@article{Leone, author={Leone, Hudson and Turner, Peter S. and Devitt, Simon}, title={Quantum circuits for simulating linear interferometers}, journal={Physical Review Research}, volume={8}, pages={013190}, year={2026}, doi={10.1103/z5l3-z2dr}}

@misc{Cuccaro,
 author={Cuccaro, Steven A. and Draper, Thomas G. and Kutin, Samuel A. and Moulton, David Petrie},
 title={A new quantum ripple-carry addition circuit},
 year={2004}, eprint={quant-ph/0410184}, archivePrefix={arXiv}}

@article{Bennett1973,
 author={Bennett, Charles H.},
 title={Logical Reversibility of Computation},
 journal={IBM Journal of Research and Development},
 volume={17}, number={6}, pages={525--532}, year={1973},
 doi={10.1147/rd.176.0525}}

@article{NishimuraOzawa,
  author={Nishimura, Harumichi and Ozawa, Masanao},
  title={Computational complexity of uniform quantum circuit families and quantum {Turing} machines},
  journal={Theoretical Computer Science},
  volume={276},
  pages={147--181},
  year={2002},
  eprint={quant-ph/9906095},
  archivePrefix={arXiv}
}

@inproceedings{Paeth1986,
 author={Paeth, Alan W.},
 title={A Fast Algorithm for General Raster Rotation},
 booktitle={Proceedings of Graphics Interface '86},
 pages={77--81},
 year={1986}
}

@article{MoshinskyQuesne,
 author={Moshinsky, Marcos and Quesne, Christiane},
 title={Linear Canonical Transformations and Their Unitary Representations},
 journal={Journal of Mathematical Physics},
 volume={12}, number={8}, pages={1772--1780}, year={1971},
 doi={10.1063/1.1665805}
}

@article{Quesada2020, author={Quesada, Nicol\'{a}s and Arrazola, Juan Miguel}, title={Exact simulation of {Gaussian} boson sampling in polynomial space and exponential time}, journal={Physical Review Research}, volume={2}, pages={023005}, year={2020}, doi={10.1103/PhysRevResearch.2.023005}}

@article{Bulmer2022, author={Bulmer, Jacob F. F. and Bell, Bryn A. and Chadwick, Rachel S. and Jones, Alex E. and Moise, Diana and Rigazzi, Alessandro and Thorbecke, Jan and Haus, Utz-Uwe and Van Vaerenbergh, Thomas and Patel, Raj B. and Walmsley, Ian A. and Laing, Anthony}, title={The boundary for quantum advantage in {Gaussian} boson sampling}, journal={Science Advances}, volume={8}, pages={eabl9236}, year={2022}, doi={10.1126/sciadv.abl9236}}

@article{OhMPO2021, author={Oh, Changhun and Noh, Kyungjoo and Fefferman, Bill and Jiang, Liang}, title={Classical simulation of lossy boson sampling using matrix product operators}, journal={Physical Review A}, volume={104}, pages={022407}, year={2021}, doi={10.1103/PhysRevA.104.022407}}

@article{LiuMPO2023, author={Liu, Minzhao and Oh, Changhun and Liu, Junyu and Jiang, Liang and Alexeev, Yuri}, title={Simulating lossy {Gaussian} boson sampling with matrix-product operators}, journal={Physical Review A}, volume={108}, pages={052604}, year={2023}, doi={10.1103/PhysRevA.108.052604}}

@article{GoPartial2025, author={Go, Byeongseon and Oh, Changhun and Jeong, Hyunseok}, title={Quantum Computational Advantage of Noisy Boson Sampling with Partially Distinguishable Photons}, journal={PRX Quantum}, volume={6}, pages={030362}, year={2025}, doi={10.1103/rflv-gc66}}

@book{BrentZimmermann,
author={Brent, Richard P. and Zimmermann, Paul},
title={Modern Computer Arithmetic},
series={Cambridge Monographs on Applied and Computational Mathematics},
volume={18},
publisher={Cambridge University Press},
year={2010}
}

@inproceedings{CleveWatrous,
 author={Cleve, Richard and Watrous, John}, title={Fast parallel circuits for the quantum {Fourier} transform}, booktitle={Proceedings of the 41st Annual IEEE Symposium on Foundations of Computer Science}, pages={526--536}, year={2000}, eprint={quant-ph/0006004}, archivePrefix={arXiv}}
\end{document}